\documentclass{sig-alternate-10pt-noperm}
\usepackage{cite}
\usepackage{relsize}
\usepackage{color}
\usepackage{url}
\usepackage{amsmath}
\usepackage{tikz}
\usepackage{wasysym}
\usepackage{bbm, dsfont}
\usepackage{paralist}
\usepackage{pictexwd}
\usepackage{xspace}
\usepackage{balance}

\usepackage[ruled,vlined,linesnumbered,commentsnumbered]{algorithm2e}
\usetikzlibrary{arrows,trees,positioning,calc}
\usetikzlibrary{automata,shapes.arrows,backgrounds,fit}
\usetikzlibrary{decorations.pathmorphing,decorations.pathreplacing}
\usetikzlibrary{shapes}
\tikzstyle{rstate}=[state,ellipse]
\tikzset{>={latex}}

\usepackage{amsthm}

\usepackage{hyperref}

\allowdisplaybreaks
\newtheorem{theorem}{Theorem}[section]
\newtheorem{corollary}[theorem]{Corollary}
\newtheorem{proposition}[theorem]{Proposition}
\newtheorem{lemma}[theorem]{Lemma}

\newtheorem{examplethm}[theorem]{Example}
\newenvironment{example}{\begin{examplethm}\em}
	{\qed\end{examplethm}}

\newcommand{\eat}[1]{}

\newenvironment{citedtheorem}[1]
{\begin{theorem}\hskip-0.2em\e{\cite{#1}}\,\,}
	{\end{theorem}}

\def\Wone{\mathrm{W[1]}}

\def\NP{\mathrm{NP}}

\newcommand{\str}[1]{\mathbf{#1}} 
\newcommand{\strs}{\str{s}} 
\newcommand{\strr}{{\str{r}}} 
\DeclareMathOperator*{\btie}{\Bowtie}
\def\join{\mathbin{\btie}}
\DeclareMathOperator*{\bjoin}{\mathlarger{\mathlarger{\mathlarger{\join}}}}

\newcommand{\sel}{\zeta^=}
\newcommand{\proj}{\pi}

\newcommand{\spnr}[1]{[\![ {#1} ]\!]} 

\newcommand{\rc}[1]{\mathsf{#1}}
\newcommand{\svars}{\rc{Vars}}

\newcommand{\spn}[1]{[#1\rangle}

\newcommand{\vop}[1]{\mathop{#1{\vdash}}}
\newcommand{\vcl}[1]{\mathbin{{\dashv}#1}}
\newcommand{\xalphabet}{\boldsymbol{\Gamma}}
\newcommand{\xalphabetv}{\xalphabet_{V}}
\newcommand{\clean}{\mathsf{clr}} 
\newcommand{\bind}[2]{\mathop{#1\{#2\}}}
\newcommand{\ror}{\vee}
\newcommand{\cdoto}{\mathop{\cdot}}

\newcommand{\rlang}{\mathcal{R}}
\newcommand{\lang}{\mathcal{L}}
\newcommand{\refl}{\mathsf{Ref}}

\newcommand{\ta}{\mathtt{a}}
\newcommand{\tb}{\mathtt{b}}
\newcommand{\tc}{\mathtt{c}}

\newcommand{\df}{:=}

\newcommand{\emptyword}{\epsilon}

\def\e#1{\emph{#1}}

\def\dl{\mathrel{{:}{-}}}

\def\set#1{\{#1\}}

\def\atoms{\mathsf{atoms}}
\def\tup#1{\mathbf{#1}}
\newcommand{\penum}{\mathbf{E}}

\newcommand{\memstat}{\mathcal{V}}
\newcommand{\memf}{\vec{c}}
\newcommand{\mems}{\vec{c}} 
\newcommand{\Memf}{\vec{C}}
\newcommand{\mo}{\mathtt{o}}
\newcommand{\mc}{\mathtt{c}}
\newcommand{\mw}{\mathtt{w}}
\newcommand{\eclos}{\mathcal{E}}

\newcommand{\visiblespace}{\text{\textvisiblespace}}
\newcommand{\Konfset}{\mathcal{K}}
\newcommand{\korder}{<_{\Konfset}}
\newcommand{\rorder}{<_{\mathsf{r}}}
\newcommand{\undefined}{\bot}
\newcommand{\minlet}{\mathsf{minLetter}}
\newcommand{\nextlet}{\mathsf{nextLetter}}
\newcommand{\minword}{\mathsf{minString}}
\newcommand{\nextword}{\mathsf{nextString}}
\newcommand{\enumalgo}{\mathsf{enumerate}}
\newcommand{\Opr}{\mathsf{O}}
\newcommand{\rep}{\varrho}

\newcommand{\RGX}{\mathsf{RGX}}
\newcommand{\VA}{\mathsf{VA}}
\newcommand{\vset}{\mathsf{set}}
\newcommand{\VAset}{\mathsf{\VA_{\vset}}}
\newcommand{\puj}{{\{\pi,\cup,\bowtie\}}} 
\newcommand{\pujd}{{\{\pi,\cup,\bowtie\setminus\}}} 
\newcommand{\core}{{\{\pi,\sel,\cup,\bowtie\}}}

\newcommand{\Ajoin}{A_{\mathsf{join}}}	
\newcommand{\Aeq}{A_{\mathsf{eq}}}

\newenvironment{repeatresult}[2]
{\vskip0.5em\par\textsc{#1 #2.}\em}
{\vskip1em}

\newenvironment{repproposition}[1]{\begin{repeatresult}{Proposition}{#1}}{\end{repeatresult}}
\newenvironment{reptheorem}[1]{\begin{repeatresult}{Theorem}{#1}}{\end{repeatresult}}
\newenvironment{replemma}[1]{\begin{repeatresult}{Lemma}{#1}}{\end{repeatresult}}

\def\alphabet{\boldsymbol{\Sigma}}

\def\wild{^*}

\def\partitle#1{\vskip0.5em\noindent\textbf{#1}\,\,}

\newenvironment{citemize}
{\begin{compactitem}}
	{\end{compactitem}}

\newenvironment{cenumerate}
{\begin{compactenum}}
	{\end{compactenum}}

\def\footnotes#1{\footnote{\small #1}}

\begin{document}
	
	\conferenceinfo{??}{??}

	\pagestyle{plain}
	\pagenumbering{arabic}
	
	\title{Joining Extractions of Regular Expressions}  
	\numberofauthors{3} 
	\author{
		\alignauthor Dominik D.~Freydenberger
		\affaddr{Bayreuth University,}\\
		\affaddr{Bayreuth, Germany} \email{ddfy@ddfy.de}
		\alignauthor  Benny Kimelfeld\\
		\affaddr{Technion}\\
		\affaddr{Haifa 32000, Israel}
		\email{bennyk@cs.technion.ac.il}
		\alignauthor Liat Peterfreund\\
		\affaddr{Technion}\\
		\affaddr{Haifa 32000, Israel} \email{liatpf@cs.technion.ac.il} }
	
	\maketitle
	\begin{abstract}
          Regular expressions with capture variables, also known as
          ``regex formulas,'' extract relations of spans (interval
          positions) from text. These relations can be further
          manipulated via Relational Algebra as studied in the context
          of document spanners, Fagin et al.'s formal framework for
          information extraction. We investigate the complexity of
          querying text by Conjunctive Queries (CQs) and Unions of CQs
          (UCQs) on top of regex formulas. We show that the lower
          bounds (NP-completeness and W[1]-hardness) from the
          relational world also hold in our setting; in particular,
          hardness hits already single-character text!  Yet, the upper
          bounds from the relational world do not carry over. Unlike
          the relational world, acyclic CQs, and even gamma-acyclic
          CQs, are hard to compute. The source of hardness is that it
          may be intractable to instantiate the relation defined by a
          regex formula, simply because it has an exponential number
          of tuples.  Yet, we are able to establish general
          upper bounds. In particular, UCQs can be evaluated with
          polynomial delay, provided that every CQ has a bounded
          number of atoms (while unions and projection can be
          arbitrary). Furthermore, UCQ evaluation is solvable with FPT
          (Fixed-Parameter Tractable) delay when the parameter is the
          size of the UCQ.
	\end{abstract}



	\section{Introduction}
	
	Information Extraction (IE) conventionally refers to the task of
	automatically extracting structured information from text.  While
	early work in the area focused largely on military
	applications~\cite{GrishmanS96}, this task is nowadays pervasive in a
	plethora of computational challenges (especially those associated with
	Big Data), including social media
	analysis~\cite{DBLP:conf/acl/BensonHB11}, healthcare
	analysis~\cite{DBLP:journals/jamia/XuSDJWD10}, customer relationship
	management~\cite{DBLP:conf/www/AjmeraANVCDD13}, information
	retrieval~\cite{DBLP:conf/www/ZhuRVL07}, machine log
	analysis~\cite{DBLP:conf/icdm/FuLWL09}, and general-purpose Knowledge
	Base
	construction~\cite{DBLP:journals/pvldb/ShinWWSZR15,DBLP:conf/www/SuchanekKW07,DBLP:journals/ai/HoffartSBW13,DBLP:conf/naacl/YatesBBCES07}.
	
	One of the major commercial systems for rule-based IE is IBM's
	SystemT\footnotes{SystemT is the IE engine of IBM BigInsights,
		available for download as Quick Start Edition Trial at
		\url{http://www-03.ibm.com/software/products/en/ibm-biginsights-for-apache-hadoop}
		as of 12/2016.} that exposes an SQL-like declarative language named
	\e{AQL} (Annotation Query Language), along with a query plan
	optimizer~\cite{DBLP:conf/icde/ReissRKZV08} and development
	tooling~\cite{DBLP:journals/pvldb/LiuCCJR10}. Conceptually, AQL
	supports a collection of ``primitive'' extractors of relations from
	text (e.g., tokenizer, dictionary lookup, regex matcher and
	part-of-speech tagger), along with an algebra for relational
	manipulation (applied to the relations obtained from the primitive
	extractors). A similar approach is adopted by
	Xlog~\cite{DBLP:conf/vldb/ShenDNR07}, where user-defined functions,
	playing the role of primitive extractors, are manipulated by
	non-recursive
	Datalog. DeepDive~\cite{DBLP:journals/pvldb/ShinWWSZR15,DBLP:journals/sigmod/SaRR0WWZ16}
	exposes a declarative language for rules and features, which are
	eventually translated into the factors of a statistical model where
	parameters (weights) are set by machine learning. There, the primitive
	extractors are defined in a scripting language, and the generated
	relations are again manipulated by relational rules.
	
	Fagin et al.~\cite{fag:doc} proposed the framework of \e{document
		spanners} (or just \e{spanners} for short) that captures the
	relational philosophy of the aforementioned systems. Intuitively, a
	spanner extracts from a document $\str s$ (which is a string over a
	finite alphabet) a relation over the spans of $\str s$. A \e{span} of
	$\str s$ represents a substring of $\str s$ that is identified by the
	start and end indices.  An example of a spanner representation is a
	\e{regex formula}: a regular expression with embedded capture
	variables that are viewed as relational attributes. A \e{regular}
	spanner is one that can be expressed in the closure of the regex
	formulas under relational operators projection, union, and join.
	
	\newcommand{\asub}{\alpha_{\mathsf{sub}}}
	As an example, it is common to apply sentence boundary detection by
	evaluating a regex formula~\cite{Walker01sentenceboundary}, which we
	denote as $\alpha_{\mathsf{sen}}[x]$. When evaluating
	$\alpha_{\mathsf{sen}}[x]$ over a string $\strs$ (representing text in
	natural language), the result is a relation with a single attribute,
	$x$, where each tuple consists of the left and right boundaries of a
	sentence. The following regular spanner detects sentences that contain
	both an address in Belgium and the substring \texttt{police} inside
	them.  It assumes a regex formula $\alpha_{\mathsf{adr}}[y,z]$ that
	extracts (annotates) spans $y$ that represent addresses (e.\,g.,
	\texttt{Place de la Nation 2, 1000 Bruxelles, Belgium}) with the
	country being $z$ (e.g., the span of \texttt{Belgium}), the regex
	formulas $\alpha_{\mathsf{blg}}[z]$ and $\alpha_{\mathsf{plc}}[w]$
	that extract spans that are tokens with the words \texttt{Belgium} and
	\texttt{police}, respectively, and the regex formula
	$\asub[y,x]$ that extracts all pairs of spans $(x,y)$
	such that $y$ is a subspan of (i.e., has boundaries within) $x$.
	\begin{align}
		\pi_{x}\big(&
		\label{eq:intro-cq}
		\alpha_{\mathsf{sen}}[x]
		\join
		\alpha_{\mathsf{adr}}[y,z]
		\join
		\asub[y,x]\\
		\notag
		&\join
		\alpha_{\mathsf{blg}}[z]
		\join
		\alpha_{\mathsf{plc}}[w]
		\join
		\asub[w,x]
		\big)
	\end{align}
	For example, $\asub[y,x]$ can be represented as the regex
	formula
	$
	\alphabet\wild\cdot\bind{x}{\alphabet\wild\cdot\bind{y}{\alphabet\wild}\cdot\alphabet\wild}\cdot\alphabet\wild
	$
	where, as we later explain in detail, $\alphabet\wild$ matches every
	possible string.  A key construct in the framework of Fagin et
	al.~\cite{fag:doc} is the \e{variable-set automaton} that was proved
	that capture precisely the expressive power of regular spanners. A
	\e{core} spanner is defined similarly to a regular spanner, but it
	also allows the string-equality selection predicate on spans (an example can be found further down in this section).  Later
	developments of the framework include the exploration of the
	complexity of static-analysis tasks on core spanners~\cite{fre:doc},
	incorporating inconsistency repairing in
	spanners~\cite{DBLP:journals/tods/FaginKRV16}, a logical
	characterization~\cite{fre:log}, and a uniform model for relational
	data and IE~\cite{DBLP:conf/webdb/NahshonPV16}.
	
	In this paper we explore the computational complexity of evaluating
	Conjunctive Queries (CQs) and more generally Unions of CQs (UCQs) in
	the (regular and core) spanner framework. Such queries are defined
	similarly to the relational-database world, with two
	differences. First, the input is not a relational database, but rather
	a string. Second, the atoms are not relational symbols, but rather
	regex formulas. Therefore, the \e{relations} on which the query is
	applied are implicitly defined as those obtained by evaluating each
	regex formula over the input string. We refer to these queries as
	\e{regex CQs} and \e{regex UCQs}, respectively.  As an example, the
	above query in~\eqref{eq:intro-cq} is a regex CQ.  Under \e{data
		complexity}, where the query is assumed fixed, query evaluation is
	always doable in polynomial time. Here we focus on \e{combined
		complexity} where both the query and the data (string) are given as
	input. One might be tempted to claim that the literature on UCQ
	evaluation to date should draw the complete picture on complexity, by
	what we refer to as the \e{canonical relational} evaluation: evaluate
	each indvidual regex formula on the input string, and compute the UCQ
	as if it were query on an ordinary relational database. There are,
	however, several problems with this claim.
	
	The first problem is that lower bounds on UCQ evaluation do not carry
	over immediately to our setting, since the input is not an arbitrary
	relational database, but rather a very specific one: every relation is
	obtained by applying a regex to the (same) input string. But, quite
	expectedly, we show that the standard lower bound of
	$\NP$-completeness~\cite{DBLP:conf/stoc/ChandraM77} for Boolean CQs
	(we well as $\Wone$-hardness for some parameters) remain in our
	setting.  The more surprising finding is that hardness holds even if
	the input string is a single character!  Yet, a more fundamental
	problem with the canonical relational approach to evaluation is that
	it may be infeasible to materialize the relation defined by a regex
	formula, as the number of tuples in that relation may be exponential
	in the size of the input. This problem provably increases the
	complexity, as we show that Boolean regex-CQ evaluation is
	$\NP$-complete even on acyclic CQs, and even on the more restricted
	\e{gamma-acyclic} CQs. In contrast, acyclic CQs (and more generally
	CQs of \e{bounded hypertree
		width}~\cite{DBLP:journals/jacm/GottlobMS09}) admit polynomial-time
	evaluation. Finally, even if we were guaranteed a polynomial bound on
	the result of each atomic regex formula, it would not necessarily mean
	that we can actually materialize the corresponding relation in
	polynomial time.
	
	Yet, in spite of the above daunting complexity we are able to
	establish some substantial upper bounds. Our upper bounds are based on
	a central algorithm that we devise in this paper for evaluating a
	variable-set automaton (\e{vset-automaton} for short) over a string.
	More formally, recall that a vset-automaton $A$ represents a spanner,
	which we denote as $\spnr{A}$. When evaluating $A$ on a string
	$\strs$, the result is a relation $\spnr{A}(\strs)$ over the spans of
	$\strs$. The number of tuples in $\spnr{A}(\strs)$ can be exponential
	in the size of the input ($\strs$ and $A$). Our algorithm takes as
	input a string $\strs$ and a vset-automaton $A$ (or, more precisely, a
	\e{functional} vset-automaton~\cite{fre:doc}) and enumerates the
	tuples of $\spnr{A}(\strs)$ with \e{polynomial
		delay}~\cite{DBLP:journals/ipl/JohnsonP88}. This is done by a
	nontrivial reduction to the problem of enumerating all the words of a
	specific length accepted by an NFA~\cite{ack:eff}.  Our central
	algorithm implies several upper bounds, which we establish in two
	approaches: (a) via what we call the \e{canonical relational
		evaluation}, and (b) via \emph{compilation to automata}.
	
	The first approach utilizes known algorithms for relational UCQ
	evaluation, by materializing the relations defined by the regex
	formulas. For that, we devise an efficient compilation of a regex
	formula into a vset-automaton, and establish that a regex formula can
	be evaluated in polynomial total time (and even polynomial delay). In
	particular, we can efficiently materialize the relations of each atom
	of a regex UCQ, whenever we have a polynomial bound on the cardinality
	of this relation. Hence, there is no need for a specialized algorithm
	for each cardinality guarantee---one algorithm fits all. Consequently,
	under such cardinality guarantees, canonical relational evaluation is
	efficient whenever the underlying UCQ is tractable (e.g., each CQ is
	acyclic).
	
	In the second approach, compilation to automata, we compile the entire
	regex UCQ into a vset-automaton. Combining our polynomial-delay
	algorithm with known results~\cite{fag:doc,fre:log}, we conclude that
	regex UCQs can be evaluated with Fixed-Parameter Tractable (FPT) delay
	when the size of the UCQ is the parameter. Moreover, we prove that the
	compilation is efficient \e{if} every disjunct (regex CQ) has a
	bounded number atoms. In particular, we show that every join of a
	bounded number of vset-automata can be compiled in polynomial time
	into a single vset-automaton, and every projection and union (with no
	bounds) over vset-automata can be compiled in polynomial time into a
	single vset-automaton. Hence, we establish that for every fixed $k$,
	the evaluation of regex $k$-UCQs (where each CQ has at most $k$ atoms)
	can be performed with polynomial delay.
	
	Finally, we generalize our results to allow regex-UCQ to include
	string-equality predicates, which are expressions of the form ``$x$
	and $y$ span the same substring, possibly in different locations''.
	For example, the following regex CQ with a string equality finds
	sentences that have the address as expressed in our previous example
	spanner (see~\eqref{eq:intro-cq} above), but they are not necessarily
	the same sentences (and in particular they may exclude the word
	\texttt{police}).
	\[\pi_{x'}
	\sel_{y,y'}
	\big(
	Q[w,x,y,z]
	\join
	\alpha_{\mathsf{sen}}[x']
	\join
	\asub[y',x']
	\big)
	\]
	where $Q[w,x,y,z]$ is the join expression inside the projection
	of~\eqref{eq:intro-cq} above.  As shown in~\cite{fre:doc}, adding an
	unbounded number of string equalities can make a trac\-table regex UCQ
	intractable, even if that regex UCQ is a single regex formula. We
	prove that now we no longer retain the above FPT delay, as the problem
	is $\Wone$-hard when the parameter is the size of the query. 
	
	Nevertheless, much of the two evaluation approaches generalize to
	string equalities. While this generalization is immediate for
	the first approach, the second faces a challenge---it is
	\e{impossible} to compile string equality into a
	vset-automaton~\cite{fag:doc}. Yet, we show that our compilation
	techniques allow to compile in string equality \e{for the specific
		input string} at hand (that is, not statically but rather at
	runtime). Hence, we conclude that regex $k$-UCQs with a bounded number
	of string equalities can be evaluated with polynomial delay.
	
	The rest of the paper is organized as follows. In
	Section~\ref{sec:preliminaries} we recall the framework of document
	spanners, and set the basic notation and terminology.  In
	Section~\ref{sec:complexity} we give our complexity results for regex
	UCQs. We describe our polynomial-delay algorithm for evaluating
	vset-automata in Section~\ref{sec:delay}. In Section~\ref{sec:streq}
	we generalize our complexity results to UCQs with string
	equalities. Finally, we conclude in Section~\ref{sec:conclusions}.

	\section{Basics of Document Spanners}\label{sec:preliminaries}
	We begin by recalling essentials of \e{document
		spanners}~\cite{fag:doc,DBLP:journals/tods/FaginKRV16,fre:doc}, a
	formal framework for Information Extraction (IE) inspired by
	declarative IE systems such as Xlog~\cite{DBLP:conf/vldb/ShenDNR07}
	and IBM's
	SystemT~\cite{DBLP:journals/sigmod/KrishnamurthyLRRVZ08,DBLP:conf/acl/LiRC11,DBLP:conf/icde/ReissRKZV08}.
	
	\subsection{Strings, Spans and Spanners}
	Throughout this paper, we fix a finite alphabet $\alphabet$. Unless
	explicitly stated, we assume that $|\alphabet|\geq 2$.  A \e{string}
	is a sequence $\sigma_1 \sigma_2 \cdots \sigma_{\ell}$ of symbols from
	$\alphabet$.  The set of all strings is denoted by
	$\alphabet\wild$. We use bold (non-italic) letters, such as $\tup s$
	and $\tup t$, to denote strings, and $\emptyword$ to denote the empty
	string. The length $\ell$ of a string $\str{s} \df \sigma_1 \sigma_2
	\cdots \sigma_{\ell}$ is denoted by $|\str s|$.
	
	Let $\str{s} \df \sigma_1 \sigma_2 \cdots \sigma_{\ell}$ be a string.
	A \e{span} of $\str{s}$ represents an interval of characters in $\str
	s$ by indicating the bounding indices. Formally, a span of $\str s$ is
	an expression of the form $\spn{i,j}$ with $1 \leq i \leq j \leq
	\ell+1$. For a span $\spn{i,j}$ of $\str{s}$, we denote by
	$\str{s}_{\spn{i,j}}$ the string $\sigma_i\cdots \sigma_{j-1}$. Note
	that two spans $\spn{i_1,j_1}$ and $\spn{i_2,j_2}$ are \e{equal} if
	and only if both $i_1=i_2$ and $j_1=j_2$ hold. In particular,
	$\strs_{\spn{i_1,j_1}}=\strs_{\spn{i_2,j_2}}$ does not necessarily
	imply that $\spn{i_1,j_1}=\spn{i_2,j_2}$.
	
	\begin{example}
		Our examples assume that $\alphabet$ consists of the Latin
		characters and some common additional symbols such as punctuation
		and commercial at ($\mathtt{@}$). In addition, the symbol
		$\visiblespace$ stands for whitespace.  
		
		Let $\strs$ be the string $\mathtt{chocolate\visiblespace
			cookie}$. Then $|\strs|=16$. The strings $\strs_{\spn{4,6}}$ and
		$\strs_{\spn{11,13}}$ are both equal (to the string $\mathtt{co}$),
		and yet, $\spn{4,6}\neq \spn{11,13}$. Likewise, we have
		$\strs_{\spn{1,1}}=\strs_{\spn{2,2}}=\emptyword$, but $\spn{1,1}\neq
		\spn{2,2}$. Finally, observe that $\strs$ is the same as
		$\strs_{\spn{1,17}}$.
	\end{example}

	We assume an infinite set $\svars$ of variables, disjoint from
	$\alphabet$. For a finite $V \subset \svars$ and $\strs \in
	\alphabet\wild$, a \emph{$(V,\strs)$-tuple} is a function $\mu$ that
	maps each variable in $V$ to a span of $\strs$. If context allows, we
	write $\strs$-tuple instead of $(V,\strs)$-tuple. A set of
	$(V,\strs)$-tuples is called a \emph{$(V,\strs)$-relation}.  A
	\emph{spanner} is a function $P$ that is a associated with a finite
	variable set $V$, denoted $\svars(P)$, and that maps every string
	$\strs \in \alphabet\wild$ to a $(V,\strs)$-relation $P(\strs)$.  
	

	A spanner $P$ is \emph{Boolean} if $\svars(P)=\emptyset$. If $P$ is
	Boolean, then either $P(\strs)=\emptyset$ or $P(\strs)$ contains only
	the empty $(\emptyset,\strs)$-tuple; we interpret these two cases as
	\e{false} and \e{true}, respectively.

	\subsection{Spanner Representations}
	This paper uses two models as basic building blocks for spanner
	representations \e{regex formulas} and  \e{vset-automata}. These can
	be understood as extensions of regular expressions and NFAs,
	respectively, with variables. Both models were introduced by Fagin et
	al.~\cite{fag:doc}, and following Freydenberger~\cite{fre:log} we
	define the semantics of these models using so-called
	\emph{ref-words}~\cite{sch:cha} (short for \emph{reference-words}).
	\subsubsection{Ref-words}\label{sec:refwords}
	For a finite variable set $V\subset \svars$, ref-words are defined
	over the extended alphabet $\alphabet\cup \xalphabetv$, where $\xalphabetv$
	consists of two symbols, $\vop{x}$ and $\vcl{x}$, for each variable
	$x\in V$. We assume that $\alphabet$ and $\xalphabetv$ are
	disjoint. Intuitively, the letters $\vop{x}$ and $\vcl{x}$ represent
	opening or closing a variable $x$. Hence, ref-words extend terminal
	strings with an encoding of variable operations. As we shall see,
	treating these variable operations as letters allows us to adapt
	techniques from automata theory.
	
	A ref-word $\strr\in(\alphabet\cup\xalphabetv)^*$ is
	\e{valid for $V$} if each variable of $V$ is opened and then closed
	exactly once, or more formally, for each $x\in\svars(A)$ the string
	$\strr$ has precisely one occurrence of $\vop{x}$, precisely one
	occurrence of $\vcl{x}$, and the former occurrence takes place before
	(i.e., to the left of) the latter.
	\begin{example}
		Let $V\df\{x\}$, and define the ref-words $\strr_1\df \tc\vop{x}\mathtt{oo}\vcl{x} \mathtt{ie}$, $\strr_2\df \vop{x}\vcl{x}$, $\strr_3\df \vcl{x}\ta\vop{x}$, and $\strr_4\df \vop{x}\ta\vcl{x}\vop{x}\ta\vcl{x}$. 
		Then $\strr_1$ and $\strr_2$ are valid for $V$, but $\strr_3$ and $\strr_4$ are not. Note that $\strr_1$ and $\strr_2$ are not valid for $V'$ with $V'\supset V$, as all variables of $V'$ must be opened and closed.
	\end{example}
	If $V$ is clear from the context, we simply say that a ref-word is
	valid.  To connect ref-words to terminal strings and later to
	spanners, we define a morphism $\clean\colon
	(\alphabet\cup\xalphabetv)\wild\to \alphabet\wild$ by $\clean(\sigma)\df
	\sigma$ for $\sigma\in\alphabet$, and $\clean(a)\df \emptyword$
	for $a\in\xalphabetv$. For $\strs\in\alphabet\wild$, let
	$\refl(\strs)$ be the set of all valid ref-words
	$\strr\in(\alphabet\cup\xalphabetv)\wild$ with $\clean(\str r)=\strs$.  By
	definition, every $\str r\in \refl(\strs)$ has a unique factorization
	$\str r = \strr_x'\cdoto \vop{x}\cdoto \strr_{x}\cdoto \vcl{x}\cdoto
	\strr_x''$ for each $x\in V$. With these factorizations, we interpret
	$\str r$ as a $(V,\strs)$-tuple $\mu^{\strr}$ by defining
	$\mu^\strr(x) \df \spn{i,j}$, where $i\df |\clean(\strr'_x)|+1$ and
	$j\df i+|\clean(\strr_{x})|$.  Intuitively, $\clean(\strr_{x})$
	contains $\strs_{\mu^\strr(x)}$, while $\clean(\strr'_{x})$ contains
	the prefix of $\strs$ to the left.
	
	An alternative way of understanding $\mu^{\strr}=\spn{i,j}$ is that
	$i$ is chosen such that $\vop{x}$ occurs between the positions in
	$\str r$ that are mapped to $\sigma_{i-1}$ and $\sigma_{i}$, and
	$\vcl{x}$ occurs between the positions that are mapped to
	$\sigma_{j-1}$ and $\sigma_{j}$ (assuming that $\strs = \sigma_1
	\cdots \sigma_{|\strs|}$, and slightly abusing the notation to avoid a
	special distinction for the non-existing positions $\sigma_{0}$ and
	$\sigma_{|\strs|+1}$).
	\begin{example}
		Let $\strs\df \mathtt{cookie}$, and define 
		$\strr_1 \df \tc\vop{x}\mathtt{oo}\vcl{x}\mathtt{kie}$, and let $\strr_2 \df
		\mathtt{cookie}\vop{x}\vcl{x}$. Let $V\df\{x\}$. Then $\strr_1,\strr_2\in\refl(\strs)$, with $\mu^{\strr_1}(x)\df\spn{2,4}$ and
		$\mu^{\strr_2}(x)\df\spn{6,7}$.
	\end{example}
	\subsubsection{Regex Formulas}
	A \emph{regex formula (over $\alphabet$)} is a regular expression that
	may also include variables (called \e{capture variables}). Formally,
	we define the syntax with the recursive rule
	\[\alpha \df
	\emptyset\mid \emptyword \mid \sigma \mid (\alpha\ror\alpha) \mid
	(\alpha\cdot\alpha) \mid \alpha\wild \mid \bind{x}{\alpha}\] where
	$\sigma\in\alphabet$ and $x\in \svars$. We add and omit parentheses
	and concatenation ($\cdot$) symbols freely, as long as the meaning
	remains clear, and use $\alpha^+$ as shorthand for
	$\alpha\cdot\alpha\wild$, as well as $\alphabet$ as shorthand for
	$\bigvee_{\sigma\in\alphabet}\sigma$.  The set of variables that occur
	in $\alpha$ is denoted by $\svars(\alpha)$.  The \e{size} of $\alpha$,
	denoted $|\alpha|$, is naturally defined as the number of symbols in
	$\alpha$.
	
	We interpret each regex formula $\alpha$ as a generator of a ref-word
	language $\rlang(\alpha)$ over the extended alphabet
	$\alphabet\cup\xalphabet_{\svars(\alpha)}$. If $\alpha$ is of the form
	$\bind{x}{\beta}$, then $\rlang(\alpha)\df
	\vop{x}\rlang(\beta)\vcl{x}$. Otherwise, $\rlang(\alpha)$ is defined
	like the language $\lang(\alpha)$ of a regular expression; for
	example, $\rlang(\alpha\cdot\beta)\df \rlang(\alpha)\cdot
	\rlang(\beta)$. We let $\refl(\alpha)$ denote the set of all ref-words
	in $\rlang(\alpha)$ that are valid for $\svars(\alpha)$; and for every
	string $\strs \in \alphabet\wild$, we define
	$\refl(\alpha,\strs)=\rlang(\alpha)\cap \refl(\strs)$. In other words,
	$\refl(\alpha,\strs)$ contains exactly those valid ref-words from
	$\rlang(\alpha)$ that $\clean$ maps to $\strs$.
	
	Finally, the spanner $\spnr{\alpha}$ is the one that defines the following
	$(\svars(\alpha),\strs)$-relation for every string
	$\strs\in\alphabet\wild$:
	\[\spnr{\alpha}(\strs)\df
	\{\mu^\strr \mid \strr\in \refl(\alpha,\strs)\}
	\]

	We say that a regex-formula $\alpha$ is \emph{functional} if
	$\rlang(\alpha)=\refl(\alpha)$; that is, every ref-word in
	$\rlang(\alpha)$ is valid.  Fagin et al.~\cite{fag:doc} gave an
	algorithm for testing whether a regex formula is functional. By
	analyzing the complexity of their construction, we conclude the
	following.
	
	\begin{citedtheorem}{fag:doc}
		Whether a regex formula $\alpha$ with $v$ variables is functional
		can be tested in $O(|\alpha|\cdot v)$ time.
	\end{citedtheorem}

	Following the convention from~\cite{fag:doc}, we assume
	that regex formulas are functional unless explicitly noted.
	
	\begin{example}
		Consider the following regex formula~$\alpha$.
		\[ \alphabet\wild \bigl( (\bind{x}{\mathtt{foo}} \alphabet\wild
		\bind{y}{\mathtt{bar}})\ror(\bind{y}{\mathtt{bar}} \alphabet\wild
		\bind{x}{\mathtt{foo}}) \bigr) \alphabet\wild\] Then
		$\svars(\alpha)=\set{x,y}$. For $\strs\in\alphabet\wild$, the $(\svars(\alpha),\strs)$-relation
		$\spnr{\alpha}(\strs)$ contains every $(\svars(\alpha),\strs)$-tuple
		$\mu$ that satisfies $\strs_{\mu(x)}=\mathtt{foo}$ and
		$\strs_{\mu(y)}=\mathtt{bar}$. Now, consider the regex formula
		$\beta$ defined as follows.
		$$\alphabet\wild \visiblespace \bind{x_{\textsf{mail}}}{ 
			\bind{x_{\textsf{user}}}{\gamma}\mathtt{@}\bind{x_{\textsf{domain}}}{\gamma
				\mathtt{.} \gamma} }\visiblespace\alphabet\wild,$$ where
		$\gamma\df (\ta \ror \cdots \ror \mathtt{z})\wild$, and, as
		previously said, $\visiblespace$ denotes whitespace. This
		regex formula identifies (simplified) email addresses, where
		the variable $x_{\textsf{mail}}$ contains the whole address,
		and $x_{\textsf{user}}$ and $x_{\textsf{domain}}$ the
		user and domain parts. Both $\alpha$ and $\beta$ are
		functional. In contrast, $\bind{x}{\ta}\bind{x}{\ta}$
		and $\bind{x}{\ta}\ror\bind{y}{\ta}$ are not functional.
	\end{example}
	
	\subsubsection{Variable-Set Automata}\label{sec:vset}
	While the current paper is mostly motivated by spanners that are constructed using regex formulas, our upper bounds for these are obtained by converting regex formulas to the following automata model.
	
	A \e{variable-set automaton} (vset-automaton for short) with variables
	from a finite set $V\subset \svars$ can be understood as an
	$\emptyword$-NFA (i.e., an NFA with epsilon transitions allowed) that
	is extended with edges labeled with variable operations $\vop{x}$ or
	$\vcl{x}$ for $x\in V$.  Formally, a \e{vset-automaton} is a tuple
	$A\df (V,Q,q_0,q_f,\delta)$, where $V$ is a finite set of variables,
	$Q$ is the set of \e{states}, $q_0,q_f\in Q$ are the \e{initial} and
	the \e{final} states, respectively, and $\delta\colon
	Q\times(\alphabet\cup\{\emptyword\}\cup\xalphabetv)\to 2^{Q}$ is the
	\e{transition function}. By $\svars(A)$ we denote the set $V$.
	
	The vset-automaton $A\df (V,Q,q_0,q_f,\delta)$ can be interpreted as a
	directed graph, where the nodes are the states, and every $q\in
	\delta(p,a)$ is represented by an edge from $p$ to $q$ with the label
	$a$. To define the semantics of $A$, we first interpret $A$ as an
	$\emptyword$-NFA over the terminal alphabet $\alphabet\cup\xalphabetv$,
	and define its \e{ref-word language} $\rlang(A)$ as the set of all
	ref-words $\str r\in (\alphabet\cup\xalphabetv)\wild$ such that some path
	from $q_0$ to $q_f$ is labeled with $\tup r$.
	
	Let $A$ be a vset-automaton. Analogously to regex formulas, we denote by $\refl(A)$ the set of all
	ref-words in $\rlang(A)$ that are valid for $V$, and define $\refl(A,\strs)$ and $\spnr{A}(\strs)$ accordingly for every $\strs\in\alphabet\wild$.  
	Likewise, we say that $A$ is \emph{functional} if
	$\refl(A)=\rlang(A)$; i.\,e., every accepting run of $A$ generates a valid ref-word.
	Two vset-automata $A_1,A_2$ are \emph{equivalent} if $\spnr{A_1}=\spnr{A_2}$. 
	\newcommand{\Afun}{A_{\mathsf{fun}}}
	\begin{example}\label{ex:funNotFun}
		Let $A$ be the following vset-automaton:  
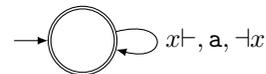
\begin{center}
			\begin{tikzpicture}[on grid, node distance =1.5cm,every loop/.style={shorten >=0pt}]
			\node[state,initial text=,initial by arrow,accepting] (q0) {};
			\path[->]
			(q0) edge[loop right] node[right] {$\vop{x},\ta,\vcl{x}$} (q0)
			;
			\end{tikzpicture}
		\end{center}
		Then, the following hold.
		\begin{align*}
			\rlang(A)=&\set{\vop{x},\ta,\vcl{x}}\wild\\
			\refl(A)=&\set{\ta^i\vop{x} \ta^j \vcl{x} \ta^k \mid i,j,k\geq 0}
		\end{align*}
		Hence, $A$ is not functional, since $\rlang(A)$ contains invalid
		ref-words such as $\emptyword$, $\vop{x}$, $\vcl{x}\ta \vop{x}$, and
		$\vop{x}\ta \vcl{x}\vcl{x}$. Now consider the vset-automaton $\Afun$:
		\begin{center}
			\begin{tikzpicture}[on grid, node distance =1.5cm,every loop/.style={shorten >=0pt}]
			\node[state,initial text=,initial by arrow] (q0) {};
			\node[state,right= of q0] (q1) {};
			\node[state,accepting,right=of q1] (q2) {};
			\path[->]
			(q0) edge[loop above] node[right,very near end] {$\ta$} (q0)
			(q0) edge node[below] {$\vop{x}$} (q1)
			(q1) edge[loop above] node[right,very near end] {$\ta$} (q1)
			(q1) edge node[below] {$\vcl{x}$} (q2)
			(q2) edge[loop above] node[right,very near end] {$\ta$} (q2)
			;
			\end{tikzpicture}
		\end{center}
		Then $\refl(\Afun)=\rlang(\Afun)=\refl(A)$. Hence, $\Afun$ is
		functional, and $\Afun$ and $A$ are equivalent. 
		
		For all $\strs\in\alphabet\wild \setminus \{\ta\}\wild$, $\spnr{A}(\strs)=\emptyset$; and for all $\strs\in \{\ta\}\wild$, $\spnr{A}(\strs)$ contains all possible $(\{x\},\strs)$-tuples.
	\end{example}
	While $\refl(A_1)=\refl(A_2)$ is a sufficient criterion for the equivalence of $A_1$ and $A_2$, it is only characteristic if the automata have at most one variable. For example, consider $\strr_1\df \vop{x}\vop{y}\vcl{x}\vcl{y}$ and $\strr_2\df \vop{y}\vcl{y}\vop{x}\vcl{x}$. Both are valid ref-words that define the same $(\{x,y\},\emptyword)$-tuple $\spn{1,1}$, but $\strr_1\neq \strr_2$.
	
	Example~\ref{ex:funNotFun} suggests that vset-automata can be
	converted into equivalent functional vset-automata; but
	Freydenberger~\cite{fre:log} showed that although this is possible
	with standard automata constructions, the resulting blow-up may be
	exponential in the number of variables.
	
	As shown in Lemma~\ref{lem:regexToAutomaton}, every functional regex formula can be converted into an equivalent functional vset-automaton; and we use this connection throughout the paper. We shall see that functional vset-auto\-ma\-ta are a convenient tool for working with regex formulas. In contrast to this,  vset-auto\-ma\-ta in general can be quite inconvenient (e.\,g., even deciding emptiness of $\spnr{A}(\emptyword)$ is NP-complete if $A$ is not functional, see~\cite{fre:log}).
	
	Freydenberger~\cite{fre:log} established the complexity of testing
	whether a given vset-automaton is functional.
	\begin{citedtheorem}{fre:log}\label{thm:funtest}
		Whether a given vset-automaton $\alpha$ with $n$ states, $m$
		transitions and $v$ variables is functional can be tested in
		$O(vm+n)$ time.
	\end{citedtheorem}
	A special property of functional vset-automata is that each state
	implicitly stores which variables have been opened and closed. We
	discuss it in detail in Section~\ref{sec:delayalgo}.

	\subsubsection{Spanner Algebras}
	Let $P$, $P_1$ and $P_2$ be spanners. The algebraic operators
	\emph{union}, \emph{projection}, \emph{natural join}, and
	\emph{selection} are defined as follows:
	
	\partitle{Union:} If $\svars(P_1) = \svars(P_2)$, their \emph{union}
	$(P_1 \cup P_2)$ is defined by $\svars(P_1 \cup P_2) \df \svars(P_1)$
	and $(P_1 \cup P_2)(\strs) \df P_1(\strs) \cup P_2(\strs)$ for all
	$\strs\in\alphabet\wild$.
	
	\partitle{Projection:} Let $Y \subseteq \svars(P)$. The
	\emph{projection} $\pi_Y P$ is defined by $\svars(\pi_Y P) \df Y$ and
	$\pi_Y P(\strs) \df {P|}_{Y}(\strs)$ for all $\strs \in
	\alphabet\wild$, where ${P|}_{Y}(\strs)$ is the restriction of all
	$\mu\in P(\strs)$ to $Y$.
	
	\partitle{Natural join:} Let $V_i \df \svars(P_i)$ for $i \in
	\{1,2\}$. The \emph{(natural) join} $(P_1 \join P_2)$ of $P_1$ and
	$P_2$ is defined by $\svars(P_1 \join P_2) \df \svars(P_1) \cup
	\svars(P_2)$ and, for all $\strs \in \alphabet\wild$, $(P_1\join
	P_2)(\strs)$ is the set of all $(V_1 \cup V_2, \strs)$-tuples~$\mu$
	for which there exist $\mu_1\in P_1(\strs)$ and $\mu_2\in P_2(\strs)$
	with ${\mu|}_{V_1}(\strs) = \mu_1(\strs)$ and ${\mu|}_{V_2}(\strs) =
	\mu_2(\strs)$.

	\partitle{Selection:} Let $R \subseteq (\alphabet\wild)^k$ be a $k$-ary
	relation over $\alphabet\wild$. The \emph{selection operator}
	$\zeta^R$ is parameterized by $k$ variables $x_1,\dots,x_k \in
	\svars(P)$, written as $\zeta^R_{x_1,\dots,x_k}$. The \emph{selection}
	$\zeta^R_{x_1,\dots,x_k} P$ is defined by
	$\svars(\zeta^R_{x_1,\dots,x_k} P) \df \svars(P)$ and, for all
	$\strs\in\alphabet\wild$, $\zeta^R_{x_1,\dots,x_k} P(\strs)$ is the
	set of all $\mu \in P(\strs)$ for which $\left(\strs_{\mu(x_1)},
	\dots, \strs_{\mu(x_k)}\right) \in R$.  Following Fagin et
	al.~\cite{fag:doc}, we almost exclusively consider binary string
	equality selections, $\sel_{x,y} P$, which selects all $\mu\in
	P(\strs)$ where $\strs_{\mu(x)}=\strs_{\mu(y)}$.
	
	\partitle{}Note that the join operator joins tuples that have identical spans in
	their shared variables. In contrast, the selection operator compares
	the substrings of $\strs$ that are described by the spans, and does
	not distinguish between different spans that span the same substrings.
	
	Following Fagin et al.~\cite{fag:doc}, we refer to regex formulas and
	vset-automata as \emph{primitive spanner representations}, and use
	$\RGX$ and $\VAset$ to denote the sets of all functional regex
	formulas and vset-automata, respectively.  A \emph{spanner algebra} is
	a finite set of spanner operators. If $\Opr$ is a spanner algebra and
	$C$ is a class of primitive spanner representations, then $C^{\Opr}$
	denotes the set of all \emph{spanner representations} that can be
	constructed by (repeated) combination of the symbols for the operators
	from $\Opr$ with spanners from $C$. For each spanner representation of
	the form $o\rep$ or $\rep_1 \mathbin{o} \rep_2$, where
	$o\in\Opr$ and $\rep,\rep_1,\rep_2\in C$, we define $\spnr{o\rep}=o\spnr{\rep}$
	and $\spnr{\rep_1 \mathbin{o} \rep_2}=\spnr{\rep_1}
	\mathbin{o} \spnr{\rep_2}$. Furthermore, $\spnr{C^{\Opr}}$ is the
	closure of $\spnr{C}$ under the spanner operators in $\Opr$.
	
	Fagin et al.~\cite{fag:doc} refer to the elements of $\RGX^{\puj}$ as
	\e{regular spanner representations}, and to the elements of
	$\RGX^{\core}$ as \emph{core spanner representations} (as these form
	the core of the query language AQL~\cite{DBLP:conf/acl/LiRC11}). One of
	the main results of~\cite{fag:doc} is that
	$\spnr{\RGX^{\puj}}=\spnr{\VAset^{\puj}}=\spnr{\VAset}$. Hence,
	$\VAset^{\core}$ has the same expressive power as core spanner
	representations. A \e{regular spanner} is a spanner that can be
	represented in a regular spanner representation, and a \e{core
		spanner} is a spanner that can be represented in a core spanner
	representation.  Fagin et al.~\cite{fag:doc} also showed the class of
	regular spanners is closed under the \e{difference} operator (i.e.,
	$\spnr{\RGX^{\puj}}=\spnr{\RGX^{\pujd}}$), but the
	class of core spanners is not.
	
	
	\subsection{(Unions of) Conjunctive Queries}
	
	In this paper, we consider Conjunctive Queries (CQs) over regex
	formulas. Such queries are defined as the class of regular spanner
	representations that can be composed out of natural join and
	projection. In addition, we explore the extension obtained by adding
	string-equality selections.
	
	More formally, a \emph{regex CQ} is an regular spanner representation
	of the form
	$q\df \pi_Y \left(\alpha_1\join\dots\join\alpha_k\right)$
	where each $\alpha_i$ is a regex formula.  A \e{regex CQ with string
		equalities} is a core spanner representation of the form
	\[q\df \pi_Y \left(\zeta^=_{x_1,x_2}\dots \zeta^=_{x_l,x_{l+1}} \left(\alpha_1\join\dots\join\alpha_k\right)\right)\]
	for some $k\geq 1$ and $l\geq 0$.
	For clarity of notation, if $Y=\{y_1,\cdots, y_m\}$, we sometimes write $q(y_1,\ldots,y_m)$ when using $q$ to make $\svars(q)$ more explicit.
	Each regex formula $\alpha_i$ and each selection $\sel_{x_j,y_j}$ is
	called an \e{atom}, where the former is a \e{regex atom} and the latter
	an \e{equality atom}. We
	denote by $\atoms(q)$ the set of atoms of $q$. For each equality atom $\sel_{x_j,y_j}$, we define $\svars(\sel_{x_j,y_j})\df\{x_j,y_j\}$.  Note that our definitions imply that every
	variable that occurs in an equality atom also appears in at least one regex atom. 
	
	In the traditional relational model (see, e.g.,~Abiteboul et
	al.~\cite{abi:fou}), a CQ is phrased over a collection of relation
	symbols (called \e{signature}), each having a predefined
	arity. Formally, a \e{relational CQ} is an expression of the form
	$Q(y_1,\dots,y_m)\dl \varphi_1,\dots,\varphi_k$ where each $y_i$ is
	variable in $\svars$ and each $\varphi_i$ is an \e{atomic relational
		formula} (or simply \e{atom}), that is, an expression of the form
	$R(x_1,\dots,x_m)$ where $R$ is an $m$-ary relation symbol and each
	$x_j$ is a variable. Similarly to regex CQs, we denote by $\atoms(Q)$
	the set of atoms of $Q$, and by $\svars(\gamma)$ the set of variables
	that occur in an atom $\gamma$.
	
	Let $q(y_1,\dots,y_m)$ be a regex CQ (with string equalities), and let
	$Q(y_1,\dots,y_m)$ be a relational CQ. We say that $q$ \e{maps to} $Q$
	if all of the following hold.
	\begin{citemize}
		\item No relation symbol occurs more than once in $Q$ (that is, $Q$
		has no self joins).  
		\item There is a bijection $\mu:\atoms(q)\rightarrow\atoms(Q)$ that
		preserves the sets of variables, that is, for each
		$\gamma\in\atoms(q)$ we have $\svars(\gamma)=\svars(\mu(\gamma))$.
	\end{citemize}
	
	Let $q$ be a regex CQ with string equalities. We say that $q$ is
	\e{acyclic} if it maps to an acyclic (or \e{alpha-acyclic}) relational
	CQ, and \e{gamma-acyclic} is it maps to a gamma-acyclic relational CQ.
	(See e.g.~\cite{abi:fou, fag:acy} for the definitions of acyclicity.)
	Recall that gamma-acyclicity is strictly more restricted than
	acyclicity (that is, every gamma-acyclic CQ is cyclic, and
	there are acyclic CQs that are not gamma-acyclic).

	A \e{Union of regex CQs}, or \e{regex UCQ} for short, is a regular spanner  $q$
	of the form $q\df \bigcup_{i=1}^{k} q_i$ where each $q_i$ is a
	regex CQ (recall that, by definition, $\svars(q_i)=\svars(q_j)$ must hold).  In a regex UCQ \e{with string equalities}, each $q_i$ can
	be a UCQ with string equalities. The following theorem follows quite
	easily from the results of Fagin et al.~\cite{fag:doc}.
	
	\newcommand{\UCQvsCoretheorem}{The following hold.
		\begin{citemize}
			\item The class of spanners expressible as regex UCQs is that of the
			regular spanners.
			\item The class of spanners expressible as regex UCQs with string
			equalities is that of the core spanners.
		\end{citemize}}
		\begin{theorem}\label{thm:UCQvsCore}
			\UCQvsCoretheorem
		\end{theorem}
		In the relational world, a \e{relational UCQ} is a query of the form
		$\bigcup_{i=1}^{l}Q_i$, where each
		$Q_i$ is a relational CQ (and $\svars(Q_i)=\svars(Q_j)$ holds).  Given a regex UCQ (with or without string equalities) $q\df \bigcup_{i=1}^{k} q_i(y_1,\dots, y_m)$  and a relational UCQ $Q\df\bigcup_{i=1}^{l}Q_i(y_1,\dots,y_m)$, we say that $q$ \e{maps to}  $Q$ if $k=l$ and
		each $q_i$ maps to $Q_i$.
		
		A family $\mathbf{P}$ of regex UCQs (with string equalities) \e{maps
			to} a family $\mathbf{Q}$ of relational UCQs if each UCQ in
		$\mathbf{P}$ maps to one or more CQ in $\mathbf{P}$.


	
\section{Complexity of UCQ Evaluation}\label{sec:complexity}
In this section we give our main complexity results for the evaluation
of regex UCQs. We remark that in this section we do not allow string
equalities; these will be discussed in Section~\ref{sec:streq}. We
begin with a description of the complexity measures we adopt.

\subsection{Complexity Measures}
Under the measure of \e{data complexity}, where the UCQ $q$ at hand is
assumed fixed (and the string $\strs$ is given as input), it is a
straightforward observation that query evaluation can be done in
polynomial time. Hence, our main measure of complexity is that of
\e{combined complexity} where both $q$ and $\strs$ are given as
input. 

The task of evaluating a regex query $q$ over an input $\str s$
requires the solver algorithm to produce all tuples in
$\spnr{\gamma(\str s)}$ for each regex formula $\gamma$ in $q$ (unless
execution is terminated before completion). In the worst case there
could be exponentially many tuples, and so \e{polynomial time} is not
a proper yardstick of efficiency. For such problems, Johnson,
Papadimitriou and Yannakakis~\cite{DBLP:journals/ipl/JohnsonP88}
introduced several complexity guarantees, which we recall here. An
\e{enumeration problem} $\penum$ is a collection of pairs $(x,Y)$
where $x$ is an \e{input} and $Y$ is a finite set of \e{answers} for
$x$, denoted by $\penum(x)$.  In our case, $x$ has the form
$(q,\strs)$ and $\penum(x)$ is $\spnr{q}(\strs)$.  A \e{solver} for an
enumeration problem $\penum$ is an algorithm that, when given an input
$x$, produces a sequence of answers such that every answer in
$\penum(x)$ is printed precisely once. We say that a solver $S$ for an
enumeration problem $\penum$ runs in \emph{polynomial total time} if
the total execution time of $S$ is polynomial in $(|x|+|\penum(x)|)$;
and in \emph{polynomial delay} if the time between every two
consecutive answers produced is polynomial in $|x|$.

We also consider \e{parameterized
	complexity}~\cite{DowneyF-FPT,Grohe-FPT-BOOK} for various parameters
determined by $q$. Formally, a \e{parameterized problem} is a decision
problem where the input consists of a pair $(x,k)$, where $x$ is an
ordinary input and $k$ is a parameter (typically small, relates to a
property of $x$). Such a problem is \e{Fixed-Parameter Tractable}
(\e{FPT}) if there is a polynomial $p$, a computable function $f$ and
a solver $S$, such that $S$ terminates in time $f(k)\cdot p(|x|)$ on
input $(x,k)$.  We similarly define \e{FPT-delay} for a parameterized
enumeration algorithm: the delay between every two consecutive answers
is bounded by $f(k)\cdot p(|x|)$. A standard lower bound is
\e{$\Wone$-hardness}, and the standard complexity assumption is that a
$\Wone$-hard problem is not FPT~\cite{Grohe-FPT-BOOK}.

Whenever we give an upper bound, it applies to a general UCQ, and
whenever we give a lower bound, it applies to Boolean CQs. When we
give asymptotic running times, we assume the unit-cost RAM-model,
where the size of each machine word is logarithmic in the size of the
input. Regarding $\alphabet$, our lower bounds and asymptotic upper
bounds assume that it is fixed with at least two characters; our
``polynomial'' upper bounds hold even if $\alphabet$ is given as part
of the input.

\subsection{Lower Bounds}
We begin with lower bounds. Recall that Boolean CQ evaluation is
$\NP$-complete~\cite{DBLP:conf/stoc/ChandraM77}. This result does not
extend directly to regex UCQs, since relations are not given directly
as input (but rather extracted from an input string using regex
formulas). But quite expectedly, the evaluation of Boolean regex CQs
indeed remains $\NP$-complete. What is less expected is that
this  holds even for strings that consist of a single
character. 

\def\thmBooleanRegexeCqNPhard{Evaluation of Boolean regex CQs is
	$\NP$-complete, and remains $\NP$-hard even under both of the following
	assumptions.
	\begin{cenumerate}
		\item Each regex formula is of bounded size
		\item The input string is of length one.
	\end{cenumerate}}
	
	\begin{theorem}
		\label{thm:BooleanRegexeCqNPhard}
		\thmBooleanRegexeCqNPhard
	\end{theorem}
	\begin{proof}
		The upper bound is obvious (even for core spanners, see~\cite{fre:doc}). For the lower bound, we construct a reduction from 3CNF-satisfia\-bi\-li\-ty to the evaluation
		problem of Boolean regex CQs.  The input to 3CNF is a formula $\psi$ with
		the free variables $x_1,\ldots, x_n$ such that $\psi$ has the form
		$C_0\wedge \cdots \wedge C_m$ where each $C_j$ is a clause.  Each
		clause is a conjunction of three literals from the set $\{x_i, \neg
		x_i \mid 1\leq i \leq n\}$. The goal is to determine whether there is an
		assignment $\tau\colon \{x_1,\ldots, x_n\} \to \{0,1\}$ that
		satisfies $\psi$.  Given a 3CNF-formula $\psi$, we construct a regex CQ
		$q$ and an input string $\strs$ such that there is a satisfying
		assignment for $\psi$ if and only if  $\spnr{q}(\strs) \ne
		\emptyset$.  We define $\strs := \mathtt{a}$. 
		To construct
		$q$, we associate each variable $x$ with a corresponding capture
		variable $x$ and each assignment $\tau$ with a regex formula that
		assigns $x$ the span $\spn{1,1}$ if $\tau(x)=0$ and $\spn{2,2}$
		if $\tau(x)=1$. For instance, given the assignment
		$\tau$ such that $\tau(x)=\tau(y)=0$ and $\tau(z)=1$ its
		corresponding regex is $x\{ y\{ \emptyword \} \} \cdot \mathtt{a}
		\cdot z \{ \emptyword \}$.  Note that since each clause $C_i$
		contains three variables, it has exactly seven satisfying assignments
		(out ot all possible eight assignments to its variables).  We denote
		these assignments by $\tau_i^1, \ldots, \tau_i^{7}$ and their
		corresponding regex formulas by $\gamma_i^1, \ldots, \gamma_i^{7}$.
		We then define $\gamma_i = \bigvee_{j=1}^7 \gamma_i^j$ and $q:=
		\proj_{\emptyset}\bjoin_{i=1}^m \gamma_i$.
		
		It is straightforward to show that if there exists a satisfying assignment $\tau$ for $\psi$ then there is at least one $\mu\in\spnr{q}(\strs)$. This $\mu$ is obtained from $\tau$ by defining  $\mu(x)\df\spn{1,1}$ if $\tau(x)=0$ and $\mu(x)\df\spn{2,2}$ if $\tau(x)=1$. 
		The other direction is shown analogously: If there is at least one  $\mu\in\spnr{q}(\strs)$, then a satisfying assignment $\tau$ for $\psi$  is obtained by defining $\tau(x)\df0$ if $\mu(x)=\spn{1,1}$ and $\tau(x)\df 1$ if $\mu(x)=\spn{2,2}$.
	\end{proof}
	
	One might be tempted to think that the evaluation of regex CQs over a string $\strs$ is
	tractable as long as the regex CQ $q\df \proj_{Y} (\alpha_1 \join \dots \join \alpha_k)$
	maps to a relational CQ of a tractable class (e.g., acyclic CQs where
	evaluation is in polynomial total
	time~\cite{DBLP:conf/vldb/Yannakakis81}), by applying what we refer to
	as the \e{canonical relational} evaluation: 
	\begin{enumerate}
		\item[(a)] Evaluate each regex formula: $r_i\df \spnr{\alpha_i}(\str
		s)$. 
		\item[(b)] Evaluate $\pi_Y(r_1\join\dots\join r_k)$ (as a relational CQ).
	\end{enumerate}
	There are, though, two problems with the canonical relational evaluation. The first
	(and main) problem is that $r_i$ may already be too large (e.g.,
	exponential number of tuples). The second problem is that, even if
	$r_i$ is of manageable size, it is not clear that it can be
	efficiently constructed. In the next section we will show that the
	second problem is solvable: we can evaluate $\alpha_i$ over $\strs$ in
	polynomial total time. However, the first problem remains. In fact,
	the following theorem states that the evaluation of regex CQs is
	intractable, even if we restrict to ones that map to acyclic CQs, and
	even the more restricted gamma-acyclic CQs!
	%
	In addition, we can show $\Wone$-hardness with respect to the number of variables or regex formulas.
	\def\thmacyclicWone{ 
		Evaluation of gamma-acyclic Boolean regex CQs is $\NP$-complete. The problem is also $\Wone$-hard with
		respect to the number of \e{(a)} variables, \e{and (b)} atoms.
	}
	\begin{theorem}
		\label{thm:aacyclicW1}
		\thmacyclicWone
	\end{theorem}
\newcommand{\enum}{m}
\newcommand{\lep}{\vop{}}
\newcommand{\rip}{\mathop{{\dashv}}}
\newcommand{\sepp}{\mathop{{\#}}}
\begin{proof}
	The $\NP$-upper bound was already discussed in the proof of Theorem~\ref{thm:BooleanRegexeCqNPhard}.
	We prove both lower bounds at the same time by defining a polynomial time FPT-reduction from the $k$-clique problem. 
	Given an undirected graph $G:=(V,E)$ and a $k\geq 2$, this problem asks whether $G$ contains a clique with $k$ nodes.
	Let  $\alphabet:=\{\ta,\tb,\lep , \sepp , \rip\}$ (the proof can  be adapted to a binary $\alphabet$ with standard  techniques). We assume  $V=\{v_1,\dots, v_n\}$, $n=|V|$, and associate each $v_i\in V$ with a unique string $\str{v}_i\in\{\ta,\tb\}^*$ such that  $|\str{v}_i|$ is $O(\log n)$.
	
	For all $1\leq i < j \leq n$, let  $\str{e}_{i,j}=\emptyword$ if $\{v_i,v_j\}\notin E$, and $\str{e}_{i,j}\df\lep\str{v}_i \sepp \str{v}_j \vcl{}$ if $\{v_i,v_j\}\in E$. Finally, we define $\strs \df (\str{e}_{1,2} \cdots \str{e}_{1,n}) \cdot (\str{e}_{2,3} \cdots \str{e}_{2,n}) \cdots (\str{e}_{n-1,n})$. Thus, $\str{s}$ encode $E$, such that an edge $\{v_i,v_j\mid i<j\}$ precedes an edge $\{v_{i'},v_{j'}\mid i'<j'\}$   if $i<i'$, or $i=i'$ and $j<j'$.
	
	Next, we construct $q$ such that $\spnr{q}(\str s)
	\ne \emptyset $ if and only if there is a $k$-clique in $G$.
	Note that a $k$-clique has $k$ nodes $v_{c(1)},\dots,v_{c(k)}$, and we assume that $i<j$ implies $c(i)<c(j)$. For each $v_{c(l)}$,  $q$ shall contain the $(k-1)$ 
	variables $y_{1,l}$ to  $y_{l-1,l}$ and $x_{l,l+1}$ to $x_{l,k}$. The intuition is that each pair $x_{i,j}$ and $y_{i,j}$ of variables corresponds to  an edge  $\{v_{c(i)}, v_{c(j)}\}$. In particular, $x_{i,j}$ and $y_{i,j}$ shall represent the node with the smaller and larger index,  respectively.
	To this end, we define 
	$\gamma:= \gamma_1\cdots \gamma_{k-1}$ with 
	$\gamma_i:= \gamma_{i,i+1}\cdots \gamma_{i,k}$ and 
	$$	\gamma_{i,j}:=
	\alphabet^* \lep \bind{x_{i,j}}{ (\ta\ror\tb)^*} \sepp  \bind{y_{i,j}}{ (\ta\ror\tb)^* } \rip\alphabet^* 
	$$
	for all $1\leq i < j \leq k$. This also reflects that the edges of the clique occur in $\strs$ in the order $\{v_{c(1)},v_{c(2)}\}$,\ldots, $\{v_{c(k-1)},v_{c(k)}\}$.	We now want to ensure that for each $1\leq l\leq k$, all $y_{i,l}$ and all $x_{l,j}$  with $1\leq i < l < j \leq k$ have to be matched to the same substring. To ensure this, for each $1\le l \le k$,  we define the regex formula
	$
	\delta_l:=  \bigvee_{i=1}^{n} \delta_{l,\str{v}_i}$
	where 
	\begin{multline*}
		\delta_{l,\str{v}} := 
		\alphabet^* \sepp  \bind{y_{1,l}}{\str{v}} \rip \alphabet^*
		\cdots 
		\alphabet^* \sepp  \bind{y_{l-1,l}}{\str{v}} \rip \alphabet^*\\
		\cdot\alphabet^* \lep \bind{x_{l,l+1}}{\str{v}} \sepp  \alphabet^*
		\cdots	
		\alphabet^* \lep \bind{x_{l,k}}{\str{v}} \sepp  \alphabet^*.
	\end{multline*}
	Finally we define $q$ to be the query 
	$$
	q\df \pi_{\emptyset} \left( \gamma \join 
	\bjoin_{1\le i \le k-1} \delta_i \right)
	$$
	
	Note that $q$ contains $O(k)$ atoms and $O(k^2)$
	variables.
	Additionally, $q$ contains no gamma-cycles since each two different $\delta_l$ have no common variables.
	Moreover, as $|\gamma|$ is $O(k^2)$, and each $|\delta_l|$ is $O(kn\log n)$, $|q|$ is $O(k^2 + k^2 n\log n)=O(k^2 n\log n)$. Furthermore,  $|\strs|$ is $O(|E|)$. Hence, $q$ and $\strs$ can be constructed in polynomial time, and the construction is FPT with respect to the number of variables and atoms.
	
	All that is left to show that the reduction is correct; that is 
	$\spnr{q}(\strs) \ne \emptyset$ iff.\ $G$ contains a $k$-clique.
	Assume that $G$ contains 
	a $k$-clique $\{v_{c(1)},\dots, v_{c(k)}\}$
	where $c(i)<c(j)$ whenever $i<j$.
	Let $\mu$ be an $\strs$-tuple that is defined as follows: For all $1\leq i < j \leq k$, find the substring $\lep \str{v}_{c(i)}\sepp \str{v}_{c(j)}\rip$. Then map $x_{i,j}$ to the span that corresponds to $\str{v}_{c(i)}$ in this substring,  and  $y_{i,j}$ to the span that corresponds to $\str{v}_{c(j)}$. Then $\mu\in\spnr{\gamma}(\strs)$, since each two nodes in the clique are connected and since the encoding of $E$ in $\strs$ is ordered. Moreover, for each $l$, the restriction of $\mu$ to $\svars(\delta_l)$ is in $\spnr{\delta_l}(\strs)$, since the strings spanned by the $y_{i,l}$ and the $x_{l,j}$ are equal, and $\mu$ respects the order of the variables in~$\delta_l$. 
	
	Now assume  $\mu\in\spnr{q}(\strs)$. We can now derive the nodes $v_{c(1)},\ldots,v_{c(k)}$ of the clique directly from the  variables $x_{i,j}$ and $y_{i,j}$, as for each $l$, $\mu\in\spnr{\delta_{l}}(\strs)$ ensures that there is a unique $c(l)$ such that $\strs_{\mu(y_{i,l})} = \strs_{\mu(x_{l,j})}=\str{v}_{c(l)}$ for all $1\leq i < l < j \leq k$. Furthermore, $\mu\in\spnr{\gamma}(\strs)$ ensures that for all $i<j$, $\mu(x_{i,j})$ and $\mu(y_{i,j})$ map to the encoding of the edge $\{v_{c(i)},v_{c(j)}\}$ in $\strs$. Hence, $\{v_{c(1)},\ldots,v_{c(k)}\}$ is a $k$-clique in $G$. 
\end{proof}
	Note that the we did not show $\Wone$-hardness with respect to the length of the query.
	
	%
	
	
	\subsection{Upper Bounds}
	Next, we give positive complexity results. We begin with the central
	result of this paper.
	
	\subsubsection{Evaluation of Variable-Set-Automata}
	Our central complexity result states the vset-automata can be
	evaluated with polynomial delay.
	
	\def \thmdelayAlgorithm {Given a functional vset-automaton $A$ with
		$n$ states and $m$ transitions, and a string $\strs$, one can
		enumerate $\spnr{A}(\str s)$ with polynomial delay
		of $O(n^2|\strs|)$, following a polynomial preprocessing of
		$O(n^2|\strs|+mn)$.}
	\begin{theorem}\label{thm:delayAlgorithm}
		\thmdelayAlgorithm
	\end{theorem} 
	We discuss the
	algorithm and other details of the proof in Section~\ref{sec:delay}.
	In the remainder of this section, we explain the implications of this theorem for both evaluation
	approaches.
	
	\subsubsection{Canonical Relational Approach}\label{sec:naiveEvaluation}
	It was already shown by Fagin et al.~\cite{fag:doc} that every regex
	formula can be converted into an equivalent vset-automaton (where a
	vset-automaton $A$ and a regex formula $\alpha$ are said to be
	\e{equivalent} if $\spnr{A} = \spnr{\alpha}$). It is probably not at
	all surprising that this is possible in a way that is efficient and
	results in functional vset-automata (recall that we assume regex
	formulas to be functional by convention).  \def\lemmaregexToAutomaton{
		Given a regex formula $\alpha$, one can construct in time
		$O(|\alpha|)$ a functional vset-auto\-ma\-ton $A$ with
		$\spnr{A}=\spnr{\alpha}$.}
	\begin{lemma}\label{lem:regexToAutomaton}
		\lemmaregexToAutomaton \end{lemma} The proof uses the 
	Thompson construction for converting regular expressions into NFAs;
	but as the proof operates via ref-words, other method for converting
	regular expressions to NFAs could be chosen as well. The same result
	was shown independently in~\cite{mor:eng}.
	
	We note here that the complexity of the preprocessing stated in
	Theorem~\ref{thm:delayAlgorithm} holds  for a general functional
	vset-automaton. If, however, the automaton is derived from a regex
	formula $\alpha$ by the construction we use for proving
	Lemma~\ref{lem:regexToAutomaton}, then the time of this preprocessing
	drops to $O(n^2 |\strs|)$, where $n$ is $O(|\alpha|)$.
	
	As a consequence, we conclude that the canonical relational approach
	to evaluation is actually efficient for UCQs, under two
	conditions. The first condition is that the regex CQs map to a
	tractable class of relational CQs. More formally, by \e{tractable
		class} of relational CQs we refer to a class $\mathbf{Q}$ of CQs
	that can be evaluated in polynomial total time,\footnotes{We could also
		use other yardsticks of enumeration efficiency, such as polynomial
		delay and \e{incremental polynomial
			time}~\cite{DBLP:journals/ipl/JohnsonP88}. Then, every occurrence
		of ``polynomial total time'' in this part would be replaced with the
		other yardstick.} such as acyclic CQs, or more generally CQs with
	\e{bounded hypertree
		width~\cite{DBLP:journals/jacm/GottlobMS09}}. But, as shown in
	Theorem~\ref{thm:aacyclicW1}, this
	condition is not enough. The second condition is that there is a
	polynomial bound on the number of tuples of each regex formula. More
	formally, we say that a class $\mathbf{A}$ of regex formulas is
	\e{polynomially bounded} if there exists a positive integer $d$ such
	that for every regex formula $\alpha\in\mathbf{A}$ and string $\str s$
	we have $|\spnr{\alpha}(\strs)|$ is $O(|\str s|^d)$.
	
	Clearly, if every regex formula in $\mathbf{A}$ can be evaluated in
	polynomial time, then $\mathbf{A}$ is polynomially bounded. From
	Theorem~\ref{thm:delayAlgorithm} we conclude that the other direction
	also hold. Hence, from Theorem~\ref{thm:delayAlgorithm} and
	Lemma~\ref{lem:regexToAutomaton} we establish the following theorem.
	
	\begin{theorem}\label{thm:boundedNumberOfResults}
		Let $\mathbf{P}$ be a class of regex UCQs. If the regex
		formulas in the UCQs of $\mathbf{P}$ belong to a polynomially
		bounded class, and
		$\mathbf{P}$ maps to a tractable class of relational UCQs, then
		UCQs in $\mathbf{P}$ can be evaluated in polynomial total time.
	\end{theorem}
	
	Examples of
	polynomially bounded classes of regex formulas are the following.
	\begin{citemize}
		\item The class of regex formulas with at most $k$ variables, for some
		fixed $k$. 
		\item The class of regex formulas $\alpha$ with a \e{key
			attribute}, that is, a variable $x\in\svars(\alpha)$ with the
		property that for all strings $\strs$ and tuples $\mu$ and $\mu$ in
		$\spnr{\alpha}(\strs)$, if $\mu(x)=\mu'(x)$ then $\mu=\mu'$.
	\end{citemize}
	A key attribute implies a polynomial bound as the number of spans of a
	string $\strs$ is quadratic in $|\strs|$. Interestingly, its existence
	can be tested in polynomial time.
	
	\def\propkeyProperty{Given a functional
		vset-automaton $A$ with $n$ states, and a variable $x\in\svars(A)$,
		it can be decided in time $O(n^4)$ whether $x$ is a key attribute. } 
	
	\begin{proposition}\label{prop:keyProperty}
		\propkeyProperty \end{proposition}
	
	
	\subsubsection{Compilation to Automata}\label{sec:compile}
	We now discuss the second evaluation approach: compiling the regex UCQ
	to a functional vset-automaton, and then applying
	Theorem~\ref{thm:delayAlgorithm}. An immediate consequence of a
	combination with past results is as follows. Fagin et
	al.~\cite{fag:doc} showed a computable conversion of spanners in a
	regular representation into a
	vset-automaton. Freydenberger~\cite{fre:log} showed a computable
	conversion of a vset-automaton into a functional
	vset-automaton. Hence, we get the following.
	
	\begin{corollary}\label{cor:fpt}
		Spanners $q$ in a regular representation (e.g., regex UCQs) can be
		evaluated with FPT delay for the parameter $|q|$.
	\end{corollary}
	
	The corollary should be contrasted with the traditional relational
	case, where Boolean CQ evaluation is $\Wone$-hard when the size of the
	query is the parameter~\cite{PapadimitriouY-DB-FPT}. Hence, in that
	respect regex UCQ evaluation is substantially more tractable than UCQ
	evaluation in the relational model.

	Our next results are established by applying \e{efficient}
	compilations. Such compilations were obtained independently by
	Morciano et al.~\cite{mor:eng}; we discuss the differences in the approaches
	after each result (generally, both here and in~\cite{mor:eng}, the
	proofs are based on standard constructions, but ref-words allow us to
	take shortcuts). Furthermore, our proofs also discuss the
	constructions with respect to Theorem~\ref{thm:delayAlgorithm}.  We
	begin with the most straightforward result, the projection operator.
	
	\def\lemprojection {Given a functional vset-automaton $A$ and
		$Y\subseteq \svars(A)$, one can construct in linear time a
		functional vset-automaton $A_Y$ with $\spnr{A_Y} = \spnr{
			\pi_Y(A)}$.} \begin{lemma}\label{lem:projection}
		\lemprojection \end{lemma} This is shown by replacing all
	transitions for operations on variables that are not in $Y$ with
	$\emptyword$-transitions. One advantage of our proof is that it
	showcases a nice use of the ref-word semantics.  The situation is
	similar for the union operator.  \def\lemunion{Given functional
		vset-automata $A_1$,\dots, $A_k$ with $\svars(A_1)=\cdots
		=\svars(A_k)$, one can construct in linear time a functional
		vset-automaton $A$ with $\spnr{A}=\spnr{A_1\cup \cdots \cup A_k}$.
	} \begin{lemma}
	\label{lem:union}
	\lemunion
\end{lemma}
Like~\cite{mor:eng}, we prove this using the union construction for NFAs. In the proof, we also argue that the upper-bound for the worst case complexity of Theorem~\ref{thm:delayAlgorithm} is lower than the number of states of the constructed automaton suggests. Observe that in Lemma~\ref{lem:union}, the number of automata is not bounded. The situation is different for the join operator, which also uses the only construction that is not completely straightforward.
\def\lemjoinVset {Given two
	functional vset-automata $A_1$ and $A_2$, each with $O(n)$
	states and $O(v)$ variables, one can construct in time
	$O(vn^4)$ a functional vset-auto\-ma\-ton $A$ with $\spnr {A_1
		\join
		A_2} =\spnr{A}$.}
\begin{lemma}\label{lem:joinVset}
	\lemjoinVset
\end{lemma}
Both this proof and the one from~\cite{mor:eng} build on the standard construction for automata intersection; the key difference is that the our employs variable configurations. 
Note that joining $k$ automata leads to a time of $O(vn^{2k})$, which is only polynomial if $k$ is bounded. Due to Theorem~\ref{thm:aacyclicW1}, this is unavoidable under standard complexity theoretic assumptions.

This motivates the following definition. Let $k\geq 0$ be fixed. A \e{regex $k$-CQ} is a regex CQ with at
most $k$ atoms. A \e{regex $k$-UCQ} is a UCQ where each CQ is a $k$-CQ
(i.\,e., a disjunction of $k$-CQs). From
Lemmas~\ref{lem:regexToAutomaton} and~\ref{lem:joinVset} we conclude
that we can convert, in polynomial time, a join of $k$ regex formulas
into a single functional vset-automaton. Then,
Lemma~\ref{lem:projection} implies that projection can also be
efficiently pushed into the functional vset-automaton. Hence, in
polynomial time we can translate a $k$-CQ into a functional
vset-automaton. Then, with Lemma~\ref{lem:union} we conclude that this
translation extends to $k$-UCQs. Finally, by applying
Theorem~\ref{thm:delayAlgorithm} we arrive at the following main result.

\begin{theorem}\label{thm:k-ucq}
	For every fixed $k$, regex $k$-UCQs can be evaluated with polynomial delay.
\end{theorem}
Hence, where Theorem~\ref{thm:k-ucq} applies, it is more powerful than Theorem~\ref{thm:boundedNumberOfResults}: The former holds for all regex $k$-UCQs, while  the latter has  additional requirements, even when limited to $k$-UCQs.

In the next section we describe the algorithm of
Theorem~\ref{thm:delayAlgorithm}. Then, in Section~\ref{sec:streq} we
extend the main theorems of this section,
Theorems~\ref{thm:boundedNumberOfResults} and~\ref{thm:k-ucq}, to
incorporate string equalities.

\section{Evaluating VSet-Automata}\label{sec:delay}
In this section, we give a high-level description of the proof of Theorem~\ref{thm:delayAlgorithm} and the resulting algorithm that, given a functional vset-automaton~$A$ and a string $\strs\in\alphabet^*$, enumerates $\spnr{A}(\strs)$ with polynomial delay. Before we proceed to the actual algorithm in Section~\ref{sec:delayalgo}, we first discuss the underlying notion of variable configurations.

\subsection{Variable Configurations}
In order to introduce the concept, we consider an arbitrary functional vset-automaton $A=(V,Q,q_0,q_f,\delta)$,  and ensure that all states are reachable from $q_0$, and that $q_f$ is reachable from every state.

Then, for each state $q\in Q$ and all $x\in V$, each ref-word $\strr\in (\xalphabetv\cup\alphabet)^*$ that takes $A$ from $q_0$ to $q$ satisfies exactly one of these mutually exclusive conditions:
\begin{cenumerate}
	\item $|\strr|_{\vop{x}}=|\strr|_{\vcl{x}}=0$,
	\item $|\strr|_{\vop{x}}=1$ and $|\strr|_{\vcl{x}}=0$, or
	\item $|\strr|_{\vop{x}}=|\strr|_{\vcl{x}}=1$, and $\vop{x}$ occurs before $\vcl{x}$.
\end{cenumerate} 
If neither of these conditions is met, $\strr$ contains $\vop{x}$ or $\vcl{x}$ twice, or the two symbols appear in the wrong order. Then we can choose any ref-word $\strr'\in(\xalphabetv\cup\alphabet)^*$ that takes $A$ from $q$ to $q_f$, and obtain a contradiction by observing that $\strr\cdot \strr'\in\rlang(A)$, although $\strr\cdot \strr'$ is not valid. Furthermore, if we compare two ref-words $\strr_1$ and $\strr_2$ that both take $A$ from $q_0$ to a common state $q$, we know that both must satisfy the same of these three conditions (as otherwise, $\strr_1 \cdot \strr'$ or $\strr_2 \cdot \strr'$ is not valid).

In other words, in each run of $A$, the information which variables have been opened or closed is stored implicitly in the states. To formalize this notion, we define the set  $\memstat$ of \emph{variable states} as $\memstat\df\{\mw,\mo, \mc\}$ (the symbols stand for $\mw$aiting, $\mo$pen, and $\mc$losed). A $\emph{variable configuration (for $V$)}$ is a function $\memf\colon V\to \memstat$ that assigns variables states to variables. For each state $q\in Q$, we define its variable configuration  $\memf_q\colon V\to \memstat$ as follows: Choose any ref-word $\strr$ that takes $A$ from $q_0$ to $q$. We define $\memf_q(x)\df \mc$ if $\strr$ contains $\vcl{x}$, $\memf_q(x)\df \mo$ if $\strr$ contains $\vop{x}$ but not $\vcl{x}$, and $\memf_q(x)\df \mw$ if $\strr$ contains neither of the two symbols (as explained above, $\memf_q$  is well-defined for each $q\in Q$). 

\begin{example}\label{ex:funConf}
	Recall the  functional vset-automaton $\Afun$ (from Example~\ref{ex:funNotFun}):
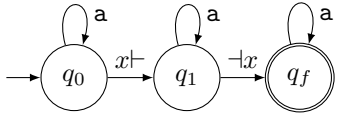
	\begin{center}
		\begin{tikzpicture}[on grid, node distance =1.5cm,every loop/.style={shorten >=0pt}]
		\node[state,initial text=,initial by arrow] (q0) {$q_0$};
		\node[state,right=of q0] (q1) {$q_1$};
		\node[state,right=of q1,accepting] (q2) {$q_f$};
		\path[->]
		(q0) edge node[above] {$\vop{x}$} (q1)
		(q1) edge node[above] {$\vcl{x}$} (q2)
		(q0) edge[loop above] node[right,near end] {$\ta$} (q0) 
		(q1) edge[loop above] node[right,near end] {$\ta$} (q1) 
		(q2) edge[loop above] node[right,near end] {$\ta$} (q2) 
		;
		\end{tikzpicture}
	\end{center}
	Then $\memf_{q_0}(x)=\mw$, $\memf_{q_1}(x)=\mo$, and $\memf_{q_f}(x)=\mc$.
\end{example}
Before we use this for the enumeration algorithm, we consider a more general view on variable configurations, that is independent of the automaton. As we shall see, the  enumeration algorithm relies on the fact that for each $\strs\in\alphabet^*$ ($\strs=\sigma_1\cdots\sigma_{\ell}$ with $\ell \geq 0$), each $(V,\strs)$-tuple can be interpreted as a sequence of $\ell+1$ variable configurations $\mems_1, \ldots, \mems_{\ell+1}$ in the following way: For $x\in V$, assume that $\mu(x)=\spn{i,j}$. For $1\leq l \leq \ell +1$, we define $\mems_l(x)\df \mw$ if $l<i$, $\mems_l(x)\df \mo$ if $i\leq l < j$, and $\mems_l(x)\df \mc$ if $l\geq j$. The idea is that each $\mems_l$ is the variable configuration immediately before reading $\sigma_l$. 

To illustrate this, consider ref-word $\strr$  with $\mu=\mu^{\strr}$ (see Section~\ref{sec:refwords}). Then $\vop{x}$ is read between $\mems_{i-1}$ and $\mems_{i}$, while $\vcl{x}$ is read between $\mems_{j-1}$ and $\mems_{j}$ (again ignoring the technicality that we do not define $\mems_0$).
\begin{example}\label{ex:tupleSeq}
	Let $V\df\{x\}$, and let $\strs\df\mathtt{aa}$. The following table contains all possible $(V,\strs)$-tuples, and the corresponding sequence $\mems_1(x),\mems_2(x),\mems_3(x)$:	
	\begin{center}
		\begin{tabular}{|l|c|c|}
			\hline 
			$\strr$ & $\mu^{\strr}(x)$ &$\memf_1(x),\memf_2(x),\memf_3(x)$ \\
			\hline 
			$\vop{x} \vcl{x} \ta\ta$ &  $\spn{1,1}$ &$\mc,\mc,\mc$ \\ 
			\hline 
			$\vop{x}\ta\vcl{x}\ta $& $\spn{1,2}$ & $\mo,\mc,\mc$\\
			\hline
			$\vop{x}\ta\ta\vcl{x} $& $\spn{1,3}$ & $\mo,\mo,\mc$\\
			\hline
			$\ta\vop{x}\vcl{x}\ta$& $\spn{2,2}$ & $\mw,\mc,\mc$\\
			\hline
			$\ta\vop{x}\ta\vcl{x}$& $\spn{2,3}$ & $\mw,\mo,\mc$\\
			\hline
			$\ta\ta\vop{x}\vcl{x}$& $\spn{3,3}$ & $\mw,\mw,\mc$\\
			\hline		
		\end{tabular} 
	\end{center}
	Note that this is exactly $\spnr{\Afun}(\strs)$, where $\Afun$ is the vset-automaton  from Example~\ref{ex:funConf}.
\end{example}
We say that a sequence of variable configurations for $V$ is \emph{valid} if it respects the order of variables states; i.\,e., $\mems_i(x)=\mc$ implies $\mems_{i+1}(x)=\mc$, and $\mems_i(x)=\mo$ implies $\mems_{i+1}(x)\in\{\mo,\mc\}$ for all $x\in V$. Obviously, each valid sequence of  $|\strs|+1$ variable configurations for $V$ can be interpreted as a $(V,\strs)$-tuple; and it is easy to see that this is a one-to-one correspondence.

To connect this point of view to the variable configurations of $A$, note that each $\strr\in\refl(A,\strs)$ can be written as $\strr = \strr_0\cdot \sigma_1\cdot \strr_1\cdot \sigma_2 \cdots \strr_{\ell-1}\cdot \sigma_{\ell}\cdot \strr_{\ell}$, where $\strr_i\in\xalphabetv^*$. For  $1\leq i \leq \ell+1$, we determine $\mems_i$ from $\strr_0 \cdot \strr_1 \cdots \strr_{i-1}$ (as in the definition of the $\memf_{q_i}$ above). This has the same effect as defining $\mems_i\df \memf_{q_i}$ for any $q_i$ that can be reached by processing $\strr_0\cdot \sigma_1\cdot \strr_1\cdot \sigma_2\cdots \strr_{i-1}$. In other words, in each run of $A$ on $\strr$, immediately before processing $\strs_i$, $A$ must be in a state $q_i$ with $\mems_i = \memf_{q_i}$.  Thus, the sequence $\mems_1,\ldots,\mems_{\ell+1}$ corresponds to the $(V,\strs)$-tuple $\mu^{\strr}$.

This observation allows us to solve the problem of enumerating $\spnr{A}(\strs)$ by enumerating the corresponding sequences of variable configurations.

Like ref-words, sequences of variable configurations can be understood as an abstraction of spanner behavior. In fact, both can be seen as successive steps of generalization: Ref-words hide the actual behavior of primitive spanner representations (i.\,e., the actual sequence of states in the vset-automaton, or which parts of the regex are mapped to which part of the input); they only express in which order variables are opened and closed. Sequences of variable configurations take this one step further, and compress successive variable operations (without terminals in-between) into a single step. This is exactly the level of granularity that is needed to distinguish different tuples.
\subsection{The Algorithm}\label{sec:delayalgo}
The algorithm enumerates the $(V,\strs)$-tuples of $\spnr{A}(\strs)$ by enumerating the corresponding sequences of $\ell+1$ variable configurations for $V$. In order to do so, we interpret each variable configuration as a letter of the alphabet $\Konfset\df\{\memf_q\mid q\in Q\}$ (note that while there are $3^{|V|}$ possible that might occur in $\Konfset$, its actual size is always bounded by $|Q|$). 
More specifically, the algorithm has the following two steps:
\begin{enumerate}
	\item Given $A$ and $\strs$, construct an NFA~$A_G$ over the alphabet $\Konfset$ such that $\lang(A_G)$ contains exactly those strings $\kappa_1\cdots \kappa_{\ell+1}$, $\kappa_i\in\Konfset$, that correspond to the elements of $\spnr{A}(\strs)$.
	\item Enumerate $\lang(A_G)$ with polynomial delay.
\end{enumerate}

The algorithm constructs the  NFA $A_G$ by first constructing a graph $G$ whose  nodes are tuples $(i,q)$, which encode that $A$ can be in state $q$ after reading $\sigma_1\cdots \sigma_i$. The edges are drawn accordingly: There is an edge from $(i,p)$ to $(i+1,q)$ if $A$ can reach $q$ from $p$ by reading $\sigma_{i+1}$ and then arbitrarily many variable operations. The NFA $A_G$ is then directly obtained from $G$ by interpreting every edge from $(i,p)$ to $(i+1,q)$ as a transition for the letter $\memf_{q}$.

We then use a modification of the algorithm by Ackerman and Shallit~\cite{ack:eff} that, given an NFA $M$ and an integer $l$, enumerates all strings in $\lang(M) \cap \alphabet^{l}$.  As $A_G$ is restricted type of NFA (all strings in $\lang(A_G)$ are of length $\ell+1$), we can ignore the intricacies of~\cite{ack:eff}; instead, we use a simplified version of the algorithm, which also reduces the running time.

In summary, the graph $G$ represents all runs of $A$ on $\strs$. By interpreting $(V,\strs)$-tuples as strings from $\Konfset^{\ell+1}$, we can treat $G$ as an NFA $A_G$, and enumerate $\spnr{A}(\strs)$ by enumerating $\lang(A_G)$. As there is a one-to-one correspondence between these two sets, the fact that $\lang(A_G)$ is enumerated without repetitions guarantees that no element of $\spnr{A}(\strs)$ is repeated.

\begin{example}\label{ex:matchNFA}
	Consider $\Afun$ (from Example~\ref{ex:funConf}) and $\strs\df \ta\ta$. From $\Afun$ and $\strs$, the algorithm constructs the NFA $A_G$ that is shown in Figure~\ref{fig:ex:matchNFA}. To see that $\lang(A_G)$ corresponds to $\spnr{\Afun}(\strs)$, take note that the table in Example~\ref{ex:tupleSeq} lists each element of $\spnr{\Afun}(\strs)$ together with its corresponding sequence of variable configurations.
	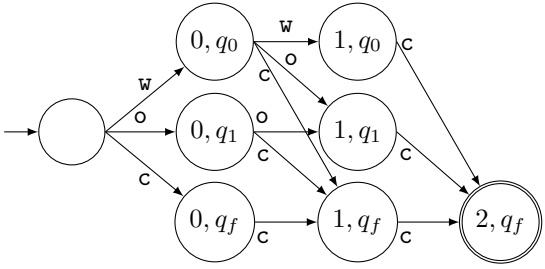
\begin{figure}
		\begin{center}
			\begin{tikzpicture}[on grid, node distance=12mm]
			\node[state] (q00) {$0,q_0$};
			\node[state, below=of q00] (q01)  {$0,q_1$};
			\node[state, below=of q01] (q02) {$0,q_f$};
			\node[state, left=of q01, xshift=-7mm,initial text=,initial by arrow] (q0) {};
			\node[state, right=of q00, xshift=7mm] (q10) {$1,q_0$};
			\node[state, below=of q10] (q11)  {$1,q_1$};
			\node[state, below=of q11] (q12) {$1,q_f$};
			\node[state, right=of q12, xshift=7mm,accepting] (q22) {$2,q_f$};
			
			\path[->]
			(q0.east) edge node[above, pos=0.5] {$\mw$} (q00)
			(q0.east) edge node[above, pos=0.5] {$\mo$} (q01)
			(q0.east) edge node[below, pos=0.5] {$\mc$} (q02) 	 	
			(q00.east) edge[above] node {$\mw$} (q10)
			(q00.east) edge[above] node {$\mo$} (q11)
			(q00.east) edge[below, very near start] node {$\mc$} (q12)
			
			(q01.east) edge[above, very near start] node {$\mo$} (q11)
			(q01.east) edge[below, very near start] node {$\mc$} (q12)
			(q02.east) edge[below, very near start] node {$\mc$} (q12)
			
			(q10.east) edge[above, very near start] node {$\mc$} (q22)	
			(q11.east) edge[below, very near start] node {$\mc$} (q22)
			(q12.east) edge[below, very near start] node {$\mc$} (q22)
			;
			\end{tikzpicture}
		\end{center}
		\caption{The NFA $A_G$ from Example~\ref{ex:matchNFA}. 	For convenience, we write each variable configuration $\memf_{q}$ as $\memf_{q}(x)$.}\label{fig:ex:matchNFA}
	\end{figure}
	A more detailed version of this example and a second example can be found in Section~\ref{app:delayExamples} in the Appendix.
\end{example}
Note that the NFA~$A_G$ in Example~\ref{ex:matchNFA} is deterministic, which means that enumerating $\lang(A_G)$ actually requires less effort than in the general case: As stated in Theorem~\ref{thm:delayAlgorithm}, if $A$ has $n$ states, we can enumerate $\spnr{A}(\strs)$ with a delay of $O(n^2|\strs|)$. But if $A_G$ is deterministic, this can be lowered to $O(|\strs|)$.

While there are automata where the worst case for the
algorithm is reached, the stated upper bounds only applies to
automata where the states have many outgoing transitions (or a
large number of $\emptyword$-transitions). For examples where
we can use it for better upper bounds, see the proofs of
Lemma~\ref{lem:union} and
Theorem~\ref{thm:enumKEqualitiesPolyDelay}.

As a final remark, we observe that the algorithm uses the variable operations on the transitions of $A$ only to compute the variable configurations. Afterwards, these transitions are treated as $\emptyword$-transitions instead. This  paper uses functional vset-automata to represent regex formulas. Instead, one could directly convert each regex formula into an $\emptyword$-NFA $A$ and a function that maps each state of $A$ to a variable configuration. The enumeration algorithm and the ``compilation lemmas'' from Section~\ref{sec:compile} would directly work with this model. Whether this actually allows constructions that are more efficient or significantly simpler remains to be seen.

\balance

	\section{String Equality}\label{sec:streq}
	Our results in the previous two sections only apply to regex UCQs. As these are not allowed to use string equality selections, these have the same expressive power as regular spanner representations.  In this section, we discuss how our results can be extended to include string equality selections, which allows us reach the full expressive power of core spanners.
	
	\subsection{Lower Bound}
	The main difficulty when dealing with string equality selections is
	that this operator quickly becomes computationally expensive, even
	without using joins. More specifically, Freydenberger and
	Holldack~\cite{fre:doc} showed that combined with string equalities, a
	single regex formula and a projection to a Boolean spanner already
	lead to an intractable evaluation problem.
	\begin{citedtheorem}{fre:doc}\label{thm:equalitiesNP}
		Evaluation of Boolean regex CQs with string equalities is
		$\NP$-complete, even if restricted to  queries of the form $\proj_{\emptyset} \sel_{x_1,y_1}\cdots \sel_{x_m, y_m} \alpha$.
	\end{citedtheorem}
	In other words, even a single regex formula already leads to $\NP$-hardness. 
	The proof
	from~\cite{fre:doc} uses a reduction from the membership problem for
	so-called pattern languages. It was shown by Fernau et
	al.~\cite{fer:par} that this membership problem is $\Wone$-complete
	for various parameters. We prove that the situation is analogous for Boolean regex
	CQs with string equalities.  \def \thmequalitiesNPWonehard
	{Evaluation of Boolean regex CQs $q$ with string equalities is
		$\Wone$-hard for the parameter $|q|$, even if restricted to queries of the form $\proj_{\emptyset} \sel_{x_1,y_1}\cdots \sel_{x_m, y_m} \alpha$.}
	\begin{theorem} \label{thm:equalitiesNPW1hard}
		\thmequalitiesNPWonehard
	\end{theorem}
	Hence, while a regex CQ that consists of a single regex formula can
	be evaluated efficiently (due to Lemmas~\ref{lem:regexToAutomaton}
	and~\ref{lem:projection} and Theorem~\ref{thm:delayAlgorithm}), even
	limited use of string equalities can become expensive. This
	result can be obtained by combining the reduction from~\cite{fer:par}
	with the reduction from~\cite{fre:doc}, the former proof requires
	encodings that are needlessly complicated for our purposes (this is
	caused by the comparatively low expressive power of pattern
	languages). Instead, we take the basic idea of an FPT-reduction from
	the $k$-clique problem, and give a direct proof that is less
	technical. In fact, the proof of Theorem~\ref{thm:equalitiesNPW1hard}
	is similar to that of Theorem~\ref{thm:aacyclicW1}, except that in
	the former the query $q$ constructed in the reduction is determined
	solely by the parameter $k$, and not the size of the graph as in the
	latter.
	
	The reader should contrast Theorem~\ref{thm:equalitiesNPW1hard} with
	Corollary~\ref{cor:fpt}. The combination of the two complexity
	results shows that, with respect to the parameter $|q|$, string
	equality considerably increases the parameterized complexity: Boolean
	regex CQ evaluation used to be FPT, and is now $\Wone$-hard.
	
	Finally, note that queries as in Theorems~\ref{thm:equalitiesNP} and~\ref{thm:equalitiesNPW1hard} are always acyclic, as the variables of each equality atom also occur in $\alpha$. Although they are not gamma acyclic, we can achieve this if we use $k$-ary string equalities, by contracting all $\sel_X$ and $\sel_Y$ with sets $X\cap Y\neq \emptyset$ into $\sel_{X\cup Y}$ until all equalities range over pairwise disjoint sets. The resulting query is equivalent, and gamma-acyclic.
	
	\subsection{Upper Bounds}
	We now examine the upper bounds for our two evaluation strategies on regex UCQs with string equalities. For the canonical relational approach, we observe that each equality atom corresponds to a relation that is of polynomial size. Hence, we can directly use  Theorem~\ref{thm:boundedNumberOfResults} to conclude the following.
	\begin{corollary}\label{}
		Let $\mathbf{P}$ be a class of regex UCQs with string equalities. 
		If the regex
		formulas in the UCQs of $\mathbf{P}$ belong to a polynomially
		bounded class, and
		$\mathbf{P}$ maps to a tractable class of relational UCQs, then
		UCQs in $\mathbf{P}$ can be evaluated in polynomial total time.
	\end{corollary}
	The upper bound for  compilation to automata requires more effort. We first observe the following.
	
	\def \thmenumKEqualitiesPolyDelay
	{For every fixed $m\ge 1$ there is an algorithm, that given a functional vset-automaton $A$, a sequence $S\df \sel_{x_1,y_1}\cdots \sel_{x_m,y_m}$ of string equality selections over $\svars(A)$, and a string $\str s$, enumerates $\spnr{S A}(\str s)$ with polynomial delay.}
	\begin{theorem}
		\label{thm:enumKEqualitiesPolyDelay}
		\thmenumKEqualitiesPolyDelay
	\end{theorem}
	The main idea of the proof is to construct a vset-automaton $\Aeq$
	that defines exactly the string equalities on $\strs$.  Specifically,
	$\mu\in\spnr{\Aeq}(\strs)$ holds if and only if 
	$\strs_{\mu(x_i)}=\strs_{\mu(y_i)}$
	for all $x_i,y_i$ where $\sel_{x_i,y_i}$ is a selection in $S$.
	We then use
	Lemma~\ref{lem:joinVset} to construct a vset-automaton $\Ajoin$ with
	$\spnr{\Ajoin}=\spnr{A\join \Aeq}$, and use
	Theorem~\ref{thm:delayAlgorithm} to enumerate
	$\spnr{\Ajoin}(\strs)=\spnr{S A}(\strs)$. Note that the construction
	of $\Aeq$ depends on $\strs$; in fact, $\spnr{\Aeq}(\strs')=\emptyset$
	for all strings $\strs'\neq \strs$. This dependency on $\strs$ is
	unavoidable: Recall that as shown by Fagin et al.~\cite{fag:doc},
	regular spanners are strictly less expressive than core spanners
	(i.\,e., string equality adds expressive power). If we could construct
	an $\Aeq$ that worked on every string $\strs'$, this would immediately
	lead to a contradiction.
	
	Analogously to the join in Section~\ref{sec:compile}, the polynomial upper bound only holds for the construction if $m$ is fixed, and Theorem~\ref{thm:equalitiesNPW1hard} suggests that this cannot be overcome under standard assumptions.
	Hence,  for all fixed $k,m\geq 0$, we define the notion of a \emph{regex $k$-UCQ with up to $m$ string equalities} analogously to a regex $k$-UCQ, with the additional requirement that each of the CQs uses at most $m$ binary string equality selections.
	It follows from Theorem~\ref{thm:enumKEqualitiesPolyDelay} that for every constant $k$, adding a fixed number of string equalities to a regex $k$-UCQ does not affect its enumeration complexity (in the sense that the complexity stays polynomial).
	\begin{corollary}
		\label{thm:enumKEqualitiesPolyDelayregexCQ}
		For all fixed $m$ and $k$, regex $k$-UCQs with up to $m$ string equalities can be evaluated  with polynomial delay.
	\end{corollary}
	In other words, Theorem~\ref{thm:k-ucq} can be extended to cover string equalities.

	
	\section{Concluding Remarks}\label{sec:conclusions}
	We have studied the combined complexity of evaluating regex CQs and
	regex UCQs. We showed that the complexity is not determined only by
	the structure of the CQ as a relational query; rather, complexity can
	go higher since an atomic regex formula can already define a relation
	that is exponentially larger than the combined size of the input and
	output. Our upper bounds are based on an algorithm for evaluating a
	vset-automaton with polynomial delay. These bounds are based on two
	altenative evaluation strategies---the {canonical relational}
	evaluation and query-to-automaton compilation. We conclude the paper
	by proposing several directions for future research.
	
	One direction is to generalize the compilation into a vset-automaton
	into a class of queries that is more general than regex $k$-UCQs with
	string equality. In particular, we would like to have a robust
	definition of a class of algebraic expressions that we can efficiently
	translate into a vset-automaton.
	
	Fagin et al.~\cite{fag:doc} have shown that regular spanners, or
	equivalently regex UCQs (as we established in
	Theorem~\ref{thm:UCQvsCore}), can be phrased as \e{Unions of
		Conjunctive Regular Path Queries}
	(UCRPQs)~\cite{Consens:1990:GVF:298514.298591,
		DBLP:conf/kr/CalvaneseGLV00,
		DBLP:conf/dbpl/DeutschT01,DBLP:journals/siamcomp/BarceloRV16}. However,
	their translation entails an exponential blowup. Moreover, even if the
	translation could be made efficiently, it is not at all clear that
	tractability properties of the regex UCQ (e.g., bounded number of
	atoms, or low hypertree width) would translate into tractable
	properties of UCRPQs. Importantly, in the UCRPQ every atom involves
	two variables, and so, the problem of intractable materialization of an
	atom (i.e., the main challenge we faced here) does not occur. So,
	another future direction is a deeper exploration of the relationship
	between the complexity results of the two frameworks.
	
	Last but not least, there is the crucial future direction of
	translating the upper bounds we presented into algorithms that
	substantially outperform the state of the art, at least when our
	tractability conditions hold. Beyond optimizing our translation and
	polynomial-delay algorithm, we would like to incorporate techniques of
	aggressive filtering for matching regular
	expressions~\cite{DBLP:journals/tods/YangQWZW016,DBLP:conf/sigmod/YangWQWL13}
	and parallelizing polynomial-delay
	enumeration~\cite{DBLP:conf/sigmod/YangWQWL13}.

	\newpage
	\onecolumn
	{
		
		\bibliographystyle{abbrv} 
		\bibliography{main}  
	}
	
	\newpage\nobalance
	\appendix

\section{Proofs}
\subsection{Proof of Theorem~\ref{thm:UCQvsCore}}

\begin{reptheorem}{\ref{thm:UCQvsCore}}
	\UCQvsCoretheorem
\end{reptheorem} 
\begin{proof}
	We begin with the first claim, that regex UCQs express exactly $\spnr{\RGX^{\puj}}$, the class of regular spanners. First, note that regex UCQs are a special case of regular spanner representations. For the other direction, 	the proof of Lemma~4.13 by Fagin et al.~\cite{fag:doc} shows that every  vset-automaton can be converted into an equivalent a union of projections over joins of regex formulas -- or, in our terminology, a UCQ.
	
	For the second claim, we also only need to proof that every core spanner can be expressed by a regex UCQ with string equalities. For this, we use the core-simplification lemma (Lemma 4.19 in~\cite{fag:doc}), which states that every core spanner $P$ can be expressed as $\proj_Y S A$, where $A$ is a vset-automaton, $S$ is a sequence of string equality selections over $V\df \svars(A)$, and $Y\subseteq V$. We then use the previous result to convert $A$ into an equivalent regex UCQ $q\df \bigcup_{i=1}q_i$, and define the regex UCQ with string equalites $q' \df \bigcup_{i=1} q'_i$, where each $q'_i$ is obtained by adding $\proj_Y$ and $S$ to $q_i$. 
	Obviously, $q'$ and $P$ are equivalent.
\end{proof}

\subsection{Proof of Theorem~\ref{thm:delayAlgorithm}}
\begin{reptheorem}{\ref{thm:delayAlgorithm}}
	\thmdelayAlgorithm
\end{reptheorem} 
\begin{proof}
	Let $\strs=\sigma_1\cdots \sigma_{\ell}$ with $\ell\df |\strs|$ and $\sigma_i\in\alphabet$.
	Let $A=(X,Q,\delta,q_0,q_f)$ be a functional vset-automaton.  Without loss of generality, we assume that each state is reachable from $q_0$, and $q_f$ can be reached from every state in $Q$.  Let $n\df |Q|$, $v\df |X|$, and let $m$ denote the number of transitions of $A$. 
	
	This proof is structured as follows: First, we consider some preliminary considerations. In particular, we discuss how every functional vset-automaton can be simulated using an NFA together with a function that represents how the vset-automaton operates on variables. Building on this, we then give the actual enumeration algorithm.
	
	\partitle{Prelimary Considerations: }
	As $A$ is functional, we know that for all $q\in Q$  and all ref-words $\strr_1,\strr_2$ with $q\in\delta(q_0,r_1)$ and $q\in\delta(q_0,r_2)$, $|\strr_1|_{a} = |\strr_2|_{a}$ holds for all $a\in\xalphabet_{X}=\{\vop{x},\vcl{x}\mid x\in X\}$. Furthermore, we know that each ref-word that is accepted by $A$ is also valid (i.e., for each $x\in X$, the ref-word contains each of $\vop{x}$ and each $\vcl{x}$ exactly once, and in the correct order). In other words, in each state $q\in Q$ and for each variable $x\in X$, it is well-defined whether $x$ has not been opened yet, or has been opened, but not closed, or has been opened and closed. We formalize this as follows: First, define the set  $\memstat$ of \emph{variable states} as $\memstat\df\{\mw,\mo, \mc\}$, where $\mw$ stands for \emph{waiting}, $\mo$ for \emph{open}, and $\mc$ for \emph{closed}. A \emph{variable configuration} is a function $\memf\colon X\to \memstat$, which maps variables to variable states. For each $q\in Q$, we choose any ref-word $\strr_q$ sucht that $q\in \delta(q_0,\strr_q)$, and use this to define $\memf_q\colon X\to \memstat$, the \emph{variable configuration of the state}~$q$, as follows:
	$$
	\memf_q(x)\df\begin{cases}
	\mw  & \text{if $\strr_q$ contains neither $\vop{x}$, nor $\vcl{x}$,}\\
	\mo  & \text{if $\strr_q$ contains  $\vop{x}$, but not $\vcl{x}$,}\\
	\mc & \text{if $\strr_q$ contains $\vop{x}$ and $\vcl{x}$}.
	\end{cases}
	$$
	Due to $A$ being functional (and due to our initial assumption that no state of $A$ is redundant), no other cases need to be considered, and $\memf_q$ is well-defined for every $q\in Q$.  We observe that $\memf_{q_0}(x)=\mw$ and $\memf_{q_f}(x)=\mc$ holds for all $x\in X$.
	
	The \emph{variable configuration of the automaton}~$A$ is the function $\Memf$ that maps each $q\in Q$ to its variable configuration $\memf_q$, which allows us to interpret $\Memf$ as a function $\Memf\colon Q\times X\to\memstat$ by $\Memf(q,x)=\memf_q(x)$. To see how $A$ operates on variables when changing from a state $p$ to a state $q$, it suffices to compare $\memf_{p}$ and $\memf_{q}$. We now use this to define an $\emptyword$-NFA $A_N$ over $\alphabet$ that, when combined with $\Memf$, simulates $A$. Let $A_N\df (Q,\delta_N,q_0,q_f)$, where $\delta_N$ is obtained from $\delta$ by defining, for all $q\in Q$,  $\delta_N(q,\sigma)\df \delta(q,\sigma)$ for all $\sigma\in\alphabet$, as well as 
	\begin{equation*}
		\delta_N(q,\emptyword)\df \delta(q,\emptyword)\cup \bigcup_{a\in\xalphabet_{X}}\delta(q,a).
	\end{equation*}
	In other words, $A_N$ is obtained from $A$ by replacing each transition that has a variable operation with an $\emptyword$-transition. Now, there is a one-to-one-correspondence between paths in $A_N$ and $A$. This also follows from the fact that $A$ is functional: In a general vset-automaton, it would be possible to have two states $p$ and $q$ such that there are two transitions with different variable operations from $p$ to $q$ (e.\,g. $\vop{x}$ and $\vcl{x}$, recall Example~\ref{ex:funNotFun}). But in functional automata, this is not possible (we already used this insight to define variable configurations).
	
	Furthermore, for every accepting run of $A_N$ on $\strs$, we can use $\Memf$ to derive a $\mu\in\spnr{A}(\strs)$ that corresponds to this run. Due to the $\emptyword$-transitions, this requires some additional effort. We begin with a definition: For each $q\in Q$, define $\eclos(q)$, the $\emptyword$-closure of $q$, as the set of all $p\in Q$ that can be reached from $q$ by using only $\emptyword$-transitions (including $q$ itself). Note that $\eclos$ can also be derived directly from $A$, by defining $\eclos(q)$ as the set of all $p$ that can be reached from $q$ by using only $\emptyword$-transitions and variable transitions. 
	
	Now, assume $\strs\in\lang(A_N)$, and recall that $\strs=\sigma_1 \cdots \sigma_{\ell}$ with $\ell\geq 0$. Then there exists a sequence of states $q_0, \hat{q}_0, \ldots, q_{\ell}, \hat{q}_{\ell}\in Q$ such that
	\begin{enumerate}
		\item $\hat{q}_i\in\eclos(q_i)$ for $0\leq i \leq \ell$,
		\item $q_{i+1}\in\delta_N(\hat{q}_i,\sigma_{i+1})$ for $0\leq i <\ell$,
		\item $\hat{q}_l=q_f$.
	\end{enumerate}
	In other words, each $q_i$ is the state that is entered immediately after reading $\strs_i$, and $\hat{q}_i$ is the state immediately before $\strs_{i+1}$ is read. Hence, $\hat{q}_{i+1}$ is reached from $\hat{q}_i$ solely using $\emptyword$-transitions (which correspond to variable operations or to $\emptyword$-transitions in $A$). Hence, in this point of view on runs, we only focus on the parts of an accepting run immediately before and after processing terminals. Take note that it is possible that $q_i=\hat{q}_i$, which corresponds to not operating on variables between reading $\sigma_i$ and $\sigma_{i+1}$. 
	
	We now discuss how this can be used to derive a $(X,\strs)$-tuple $\mu$. When introducing variable configurations, we already remarked that $\memf_{q_0}(x)=\mw$ and $\memf_{q_f}(x)=\mc$ must hold for all $x\in X$. The latter immediately implies $\memf_{\hat{q}_n}=\mc$. Furthermore, as terminal transitions cannot change variable configurations, $\memf_{\hat{q}_i}=\memf_{q_{i+1}}$ must hold as well. This allows us to define $\mu$ solely by considering the $\memf_{\hat{q}_i}$. As $A$ is functional, for each $x\in X$, there exist $i$ and $j$ with $0\leq i\leq j  \leq \ell$ such that $\memf_{\hat{q}_{i'}}=\mw$ for all $i'<i$ and $\memf_{\hat{q}_{j'}}=\mc$ for all $j'\geq j$. Less formally: In the variable  configurations of the states $\hat{q}_0, \hat{q}_1, \ldots, \hat{q}_n$, there is first a (possibly empty) block of states where $x$ is waiting, followed by a (possibly empty) block where $x$ is open, and finally a non-empty block where $x$ is closed. (The block of open configurations is empty if and only if $x$ is opened and closed without reading a terminal letter). Hence, for every $x\in X$, there exist
	\begin{enumerate}
		\item a minimal $i$, such that $\memf_{\hat{q}_i}\neq \mw$,
		\item a minimal $j$, such that $\memf_{\hat{q}_j} = \mc$. 
	\end{enumerate} 
	We then define $\mu(x)\df \spn{i+1,j+1}$. The intuition behind this is as follows: The choice of $i$ means that in $A$, the operation $\vop{x}$ occurs somewhere between $q_{i}$ and $\hat{q}_i$. Hence, the first position in $\mu(x)$ must belong to the next terminal, which is $\strs_{i+1}$. Likewise, the choice of $j$ means that in $A$, the operation $\vcl{x}$ occurs somewhere between $q_j$ and $\hat{q}_j$, which means that $\strs_{j+1}$ must be the first letter after $\mu(x)$.
	
	Thus, together with $\Memf$, every accepting run of $A_N$ for a string $\strs$ can be interpreted as a $\mu\in\spnr{A}(\strs)$. As every accepting run in $A$ corresponds to an accepting run of $A_N$, in order to enumerate $\spnr{A}(\strs)$, it suffices to enumerate all accepting runs of $A_N$. Take particular note that, instead of considering the states $\hat{q}_i$, it suffices to consider their variable configurations $\memf_{\hat{q}_i}$ in order to define $\mu$. Hence, a $(X,\strs)$-tuple $\mu$ can be derived from  a sequence of $|\strs|+1$ variable configurations. In fact, this approach also works in the reverse direction: Every $(X,\strs)$-tuple $\mu$ can be used to derive such a sequence, simply be examining how the variables are opened and closed with respect to the positions of $\strs$.
	
	\partitle{Enumeration Algorithm:} 
	The algorithm for the enumeration of $\spnr{A}(\strs)$ has two steps: First, we use the NFA $A_N$ to construct a directed acyclic graph $G$ that can be interpreted as an NFA $A_G$ such that $\lang(A_G)$ encodes exactly the accepting runs of $A$ on $\strs$. Second, we use the algorithm $\enumalgo$ (Algorithm~\ref{algo:enum} below) to enumerate all elements of $\lang(A_G)$, and thereby all of $\spnr{A}(\strs)$.  Take particular note that the alphabet of $A_G$ is the set of all $\memf_q$ with $q\in Q$, and not $\alphabet$.
	
	We begin with the first step, computing the function $\Memf$, which maps each $q\in Q$ to its variable configuration $\memf_{q}$. As $A$ is functional, all ref-words that label paths which lead from $q_0$ to a state $q$ contain exactly the same symbols from $\xalphabet_{X}$. Hence, $\Memf$ can be computed by  executing a breadth-first-search that, whenever it proceeds from a state $p$ to a state $q$, derives $\memf_{q}$ from $\memf_{p}$ (if the transition opens or closes a variable $x$, $\memf_{p}(x)$ is set to $\mo$ or $\mc$, respectively; otherwise, both variable configurations are identical). As breadth-first-search is possible in $O(m+n)$, this step runs in time $O(vm+vn)$. Note that this process can also be used to check whether $A$ is functional (this is the idea of the proof of Theorem~\ref{thm:funtest}, cf.~\cite{fre:log}).
	
	Next, we compute the $\emptyword$-closure as a function $\eclos\colon Q\to 2^{Q}$ directly from $A$. Using a standard transitive closure algorithm (see e.\,g.\ Skiena~\cite{ski:alg}),  $\eclos$ can be computed in time $O(n(n+m))$, which we can assume to be $O(mn)$. As $v \leq n \leq m$, this dominates the complexity of computing the variable configurations.
	
	Instead of directly computing $G$, we first construct a graph $G'\df(V',E')$ from $A_N$, where $V'\df \bigcup_{i=0}^{\ell} V'_i$ with
	\begin{align*}
		V'_0&\df \{(0,q)\mid q\in\eclos(q_0)\},\\
		V'_{i+1}&\df \{ (i+1,q) \mid q\in \eclos(\delta_N(p,\sigma_{i+1})) \text{ for some } (i,p)\in V'_i    \}
	\end{align*}
	for all $0\leq i < n$.
	In a slight abuse of notation, we restrict $V'_{\ell}$ to $V'_{\ell} \cap \{(\ell,q_f)\}$. Furthermore, without loss of generality, we can assume that $V'_{\ell}=\{(\ell,q_f)\}$, as otherwise, $\spnr{A}(\strs)=\emptyset$. Less formally explained, $V'_i$ contains all $(i,q)$ such that $q$ can act as $\hat{q}_i$ when using $A_N$ and $\Memf$ to simulate $A$ as explained above. 
	
	Following this idea, we define $E'$ to simulate these transitions; in other words, we define $E'\df \bigcup_{0\leq i < \ell} E_i$, where
	$$E'_i\df \{\bigl((i,p),(i+1,q)\bigr) \mid (i,p)\in V'_i \text{ and } q\in\eclos(\delta(p,\sigma_{i+1}))\}.$$
	Observe that $|V'|\leq \ell n +1$, as $|V'_i|\leq \ell$ for $0\leq i<\ell$. Also, as each node of some $V'_i$ with $i<\ell$ has at most $n$ outgoing edges, $|E'|\leq \ell n^2$.  Next, we obtain $G\df(V,E)$ by removing from $G'$ all nodes (and associated edges) from which $(n,q_f)$ cannot be reached. Using standard reachability algorithms, this is possible in time $O(|V'|+|E'|)=O(\ell n^2)$. For each $V'_i$, we define $V_i\df V'_i\cap V$.
	
	Now, every path $\pi$ from a node $(0,q)$ to $(\ell,q_f)$ in  $G$ corresponds to an accepting run of $A_N$ (and for every accepting run of $A_N$, there is a corresponding path in $G$).  Hence, if we consider a path $\pi= ( (0,\hat{q}_0), (1,\hat{q}_1),\ldots,(\ell,\hat{q}_{\ell}))$ in $G$ (which implies $\hat{q}_{\ell}=q_{\ell}$), we can use the sequence of variable configurations $\memf_{\hat{q}_0}, \memf_{\hat{q}_1},\ldots,\memf_{\hat{q}_{\ell}}$ to define a $\mu^{\pi}\in\spnr{A}(\strs)$. 
	
	While this allows us to enumerate all elements of $\spnr{A}(\strs)$, two different paths $\pi_1$ and $\pi_2$ in $G$ might lead to the same sequence of variable configurations (which would imply $\mu^{\pi_1}=\mu^{\pi_2}$), which means that $G$ alone cannot be used to compute $\spnr{A}(\strs)$ with polynomial delay. 
	
	The crucial part of the next step is interpreting the set of variable configurations of states as an alphabet $\Konfset=\{\memf_{q}\mid q\in Q\}$, and treating sequences of variable configurations as strings over $\Konfset$. Then $G$ defines the language of all variable configurations that correspond to an element of $\spnr{A}(\strs)$. By adding a starting state $q_0$, we turn $G$ into an NFA $A_{G}$ over $\Konfset$ that accepts exactly this language. More specifically, we define $A_G\df(Q_G,\delta_G,q_0,F_G)$, where
	\begin{itemize}
		\item $Q_G\df V \cup\{q_0\}$,
		\item $F_G\df V_{\ell} = \{(\ell,q_f)\}$,
		\item $\delta_G$ is defined in the following way:
		\begin{itemize}
			\item for each $(0,q)\in V_0$, $\delta_G$ has a transition from $q_0$ to $q$ with label $\memf_{q}$,
			\item for each transition $((i,p), (i+1,q))\in E$, $\delta_G$ has a transition from $p$ to $q$ with label $\memf_{q}$.
		\end{itemize}
	\end{itemize}
	We observe that all incoming transitions of a state $(i,q)$ are labeled with the same terminal letter $\memf_{q}$. Furthermore, although $|X|$ variables allow for $3^{|X|}$ distinct variable configurations, only at most $|Q|$ of these appear in $A_G$. 
	
	As there is a one-to-one correspondence between strings of $\lang(A_G)$ and elements of $\spnr{A}(\strs)$, we now want to enumerate all elements of $\lang(A_G)$. To do so, we use a variant of the general framework by Ackerman and Shallit~\cite{ack:eff} for algorithms that, given an NFA, enumerate all strings in its language that are of length~$\ell$. Note that the restricted form of $A_G$ allows us to ignore the intricacies of the algorithms that are compared in~\cite{ack:eff}, which also leads to a lower complexity.
	
	Before we proceed to the actual enumeration algorithm, we introduce another definition. 
	We fix a total order $\korder$ on $\Konfset$, and extend this to the \emph{radix order} $\rorder$ on $\Konfset^*$ as follows: Given $\str{u},\str{v}\in\Konfset^*$, we define $\str{u}\rorder \str{v}$ if $|\str{u}|<|\str{v}|$, or if $|\str{u}|=|\str{v}|$, and there exist $\str{p},\str{s}_\str{u},\str{s}_\str{v}\in \Konfset^*$ and $\kappa_{\str{u}}, \kappa_{\str{v}}\in\Konfset$ with $\str{u}= \str{p}\cdot \kappa_{\str{u}}\cdot \strs_{\str{u}}$, $\str{v}= \str{p} \cdot \kappa_{\str{v}}\cdot \strs_{\str{v}}$, and $\kappa_{\str{u}} \korder \kappa_{\str{v}}$.

	As all strings in $\lang(A_G)$ have length $\ell+1$, we write them as $\kappa_0 \kappa_1 \cdots \kappa_{\ell}$, with $\kappa_i \in \Konfset$. These indices correspond to the level numbers that are encoded in the states of $Q_G$, i.\,e., a letter $\kappa_i$ takes a state $(i-1,p)$ to a state $(i,q)$. Consequently, $\kappa_{\ell} = \memf_{q_f}$ for all strings of $\lang(A_G)$. 
	
	The enumeration algorithm uses a global stack that stores sets $S_i \subset Q_G$, which we call the \emph{state stack}. In fact, we shall see that  $S_{i+1} \subseteq \{(i,q) \mid q\in Q\}$ shall hold. Intuitively, if the algorithm has constructed a string $\kappa_0 \cdots \kappa_i$, the state stack shall store sets $S_0, \ldots, S_{i+1}$ such that each $S_{j+1}$ contains the states that can be reached from the states of $S_j$ by processing the letter $\kappa_j$. 
	As every state of $A_G$ is connected to the final state, and as $\kappa_{\ell}=\memf_{q_f}$ must hold, the algorithm does not put $S_{\ell-1}$ or $S_{\ell}$ on the state stack.
	
	Let $\undefined\notin\Konfset$ be a new letter, and define $\kappa\korder \undefined$ for all $\kappa\in \Konfset$. The enumeration algorithm uses the functions $\minlet\colon Q_G\to \Konfset$ and $\nextlet\colon Q_G \times \Konfset\to \Konfset\cup\{\undefined\}$ as subroutines, which we define as follows for all $q\in Q_G$ and $c\in\Konfset$: 
	\begin{itemize}
		\item $\minlet(q)$ is the $\korder$-smallest $\kappa\in \Konfset$ with $\delta_G(q,c)\neq \emptyset$,
		\item $\nextlet(q,\kappa)$ is the smallest $\kappa'\in\Konfset$ with $\kappa\korder \kappa'$, and $\delta_G(q,\kappa')\neq \emptyset$; or $\undefined$, if no such $\kappa'$ exists.
	\end{itemize}
	It is easily seen that these functions can be precomputed in time $O(n^2 \ell)$, ideally when computing $A_G$. The actual enumeration algorithm, $\enumalgo$, is given as Algorithm~\ref{algo:enum} below.  
	
	\begin{algorithm}[h!]
		\caption{$\enumalgo$}\label{algo:enum}
		\SetKw{KwEnumOut}{output} 
		$S_0 \df \{q_0\}$\;
		$\str{k} = \minword(0)$\;
		\While{$\str{k}\neq \undefined$}
		{
			$\KwEnumOut$ $\str{k}$\;
			$\str{k} = \nextword(\str{k})$\;
		}
	\end{algorithm}
	
	The main idea is very simple: $\enumalgo$ calls $\minword$  (Algorithm~\ref{algo:minword}) to construct the $\rorder$-smallest string of $\lang(A_G)$, using the $\minlet$ functions. Starting at $S_0=\{q_0\}$, $\minword$ determines $\kappa_0 \cdots \kappa_{\ell}$ by first chosing $\kappa_0$ as the $\korder$-smallest letter that can be read when in $q_0$, and then computing $S_1$ as all states that can be reached from there by reading $\kappa_0$. For each $S_i$, the algorithm then applies $\nextlet$ to each state of $S_i$, chooses $\kappa_i$ as the $\korder$-smallest of these letters, and puts all states that can be reached by reading $\kappa_i$ when in a state in $S_i$  into the set $S_{i+1}$. In other words, constructing the sets $S_i$ can be understood as an on-the-fly simulation of the power-set construction on $A_G$. 
	
	As all strings of $\lang(A_G)$ have length $\ell+1$, and the final state of $A_G$ is reachable from every state, this is sufficient to determine the $\rorder$-minimal string of the language. 
	
	To enumerate all further strings, $\enumalgo$ uses the subroutine $\nextword$ (Algorithm~\ref{algo:nextword} below) repeatedly. Given a string $\str{k}=\kappa_0 \cdots \kappa_{\ell}$, $\nextword$ finds the rightmost letter of $\str{k}$ that can be increased (according to $\korder$), and changes this $\kappa_i$ accordingly. It then calls $\minword$ to complete the string $\rorder$-minimally, by finding letters $\kappa_{i+1}, \ldots, \kappa_{\ell}$ that are $\korder$-minimal. While doing so, $\minword$ updates the state stack, according to the computed letters.

	\begin{algorithm}[h!]
		\caption{$\minword(l)$}\label{algo:minword}
		\SetKwInOut{Input}{Input}
		\SetKwInOut{Output}{Output}
		\SetKwInOut{Assumptions}{Assumptions}
		\SetKwInOut{Sides}{Side effects}
		\Input{an integer $l$ with $0\leq l \leq \ell$}
		\Assumptions{state stack contains $S_0, \ldots, S_{l}$ for some $\str{p}\in\Konfset^l$}
		\Output{the $\rorder$-smallest $\str{k}\in \Konfset^{\ell-l}$ that is accepted by $A_G$ when starting in a state of $S_l$}
		\Sides{updates the state stack to $S_0, \ldots, S_{\ell-1}$ for $\str{p}\cdot \str{k}$}
		\For{$i\df l$ \KwTo $\ell-1$}
		{
			find $q\in S_i$ such that $\minlet(q)$ is $\korder$-minimal\;
			$\kappa_i \df \minlet(q)$\;
			\If{$i<\ell -1$}
			{
				$S_{i+1} \df \bigcup_{p\in S_i} \delta_G(p,\kappa_i)$\;
				push $S_{i+1}$ on state stack\;
			}
		}
		$\kappa_{\ell} = \memf_{q_f}$\;
		\KwRet $\kappa_l \cdots \kappa_{\ell-1} \cdot \kappa_{\ell}$\;
	\end{algorithm}
	
	\begin{algorithm}[h!]
		\caption{$\nextword(w)$}\label{algo:nextword}
		\SetKwInOut{Input}{Input}
		\SetKwInOut{Output}{Output}
		\SetKwInOut{Assumptions}{Assumptions}
		\SetKwInOut{Sides}{Side effects}
		\Input{a string $\str{k}=\kappa_0 \cdots \kappa_{\ell}$, $\kappa_i\in \Konfset$}
		\Assumptions{state stack contains $S_0, \ldots, S_{\ell-1}$} for $\str{k}$
		\Output{the $\rorder$-smallest word $\str{k}^\prime\in\lang(A_G)$ with $\str{k}\rorder \str{k}^\prime$, or $\undefined$ if no such $\str{k}^{\prime}$ exists}
		\Sides{updates the state stack to  $S'_0, \ldots, S'_{\ell-1}$ for $\str{k'}$}
		
		\For{$i\df \ell-1$ \KwTo $0$}
		{
			let $\kappa_i$ be the $\korder$-minimal element of $\{\nextlet(q,\kappa_i)\mid q\in S_i\}$\;
			\lIf{$\kappa_i=\undefined$}{pop $S_i$ from the state stack}
			\Else{
				\If{$i<\ell -1$}
				{
					$S_{i+1} \df \bigcup_{p\in S_i} \delta_G(p,\kappa_i)$\;
					push $S_{i+1}$ on state stack\;
				}
				\KwRet $\kappa_0 \cdots \kappa_i \cdot \minword(i+1)$\;
			}
		}
		\KwRet $\undefined$\;
	\end{algorithm}
	
	We now determine the complexity of the sub-routines. We begin with $\minword$: The for-loop is executed $O(\ell)$ times. In each iteration of the loop, $q$ can be determined in time $O(n)$ by using the precomputed $\minlet$-function, and $\kappa_i$ can then be computed in $O(1)$. To build the set $S_{i+1}$, we need time $O(n^2)$. Hence, a call of $\minword$ takes at most time $O(\ell n^2)$.
	
	In each call of $\nextword$, the for-loop is executed $O(\ell)$ times. In each iteration, $\kappa_i$ can be found in time $O(n)$, as $\nextlet$ was precomputed. Furthermore, if $\kappa_i\neq \undefined$, the subroutine updates the state stack in $O(n^2)$ and terminates after a single call of $\minword$, which takes $O(\ell n^2)$. Hence, the total complexity of $\nextword$ is $O(\ell n + n^2+ \ell n^2 )$, or simply $O(\ell n^2)$.

	The total complexity of the preprocessing is obtained by adding up  the complexity of computing $\eclos$ (and $\Memf$) of $O(mn)$, constructing $G$ and $A_G$ in  $O(\ell n^2)$, computing $\minlet$ and $\nextlet$ in $O(\ell n^2)$, and a call of $\minword$ in  $O(\ell n^2)$. This adds up to  $O(\ell n^2+ mn)$ for the preprocessing and finding the first answer. Furthermore, as each call of $\nextword$ takes $O(\ell n^2)$, the delay is clearly polynomial.
	
	Recall that by Lemma~\ref{lem:regexToAutomaton}, $m$ is in $O(n)$ if we derived $A$ from a regex formula. 
	Hence, for regex formulas, the preprocessing is in  $O(\ell n^2)$ as well.
	
	Finally, note that if $A_G$ is deterministic, the complexity of the delay drops to $O(\ell)$: As each $S_i$ can contain at most state (which can have at most one letter), both $\minword$ and $\nextword$ can be computed in $O(\ell)$.
\end{proof}

\subsection{Examples for Theorem~\ref{thm:delayAlgorithm}}\label{app:delayExamples}
\begin{example}
	Consider the following vset-automaton $A$ for the functional regex formula $\ta^* \bind{x}{\ta^*} \ta^*$, where $q_0=0$ and $q_f=2$:	
	
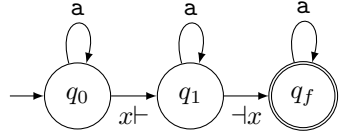
\begin{center}
		\begin{tikzpicture}[on grid, node distance =1.5cm,every loop/.style={shorten >=0pt}]
		\node[state,initial text=,initial by arrow] (q0) {$q_0$};
		\node[state,right=of q0] (q1) {$q_1$};
		\node[state,right=of q1,accepting] (q2) {$q_f$};
		\path[->]
		(q0) edge node[below] {$\vop{x}$} (q1)
		(q1) edge node[below] {$\vcl{x}$} (q2)
		(q0) edge[loop above] node {$\ta$} (q0) 
		(q1) edge[loop above] node {$\ta$} (q1) 
		(q2) edge[loop above] node {$\ta$} (q2) 
		;
		
		\end{tikzpicture}
	\end{center}
	
	Then $A$ is functional, and the variable configurations are defined as $\memf_{0}(x)=\mw$, $\memf_{1}(x)=\mo$, and $\memf_{2}(x)=\mc$. The NFA $A_N$ is obtained by replacing the each of the labels $\vop{x}$ and $\vcl{x}$ with $\emptyword$.  We now consider the input word $\strs=\ta\ta\ta$, and construct the corresponding graph $G$:
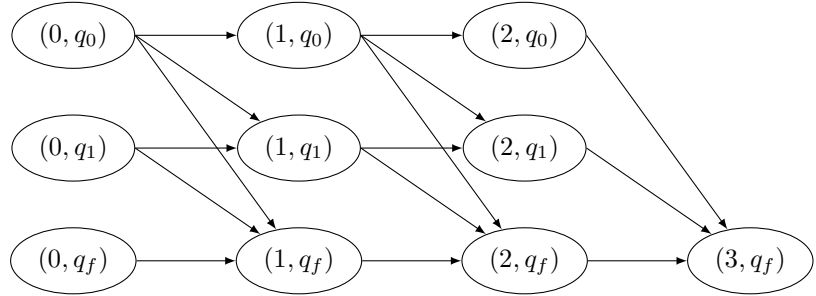
	\begin{center}
		\begin{tikzpicture}[on grid, node distance=15mm,every loop/.style={shorten >=0pt}]
		\node[rstate] (q00) {$(0,q_0)$};
		\node[rstate, below=of q00] (q01)  {$(0,q_1)$};
		\node[rstate, below=of q01] (q02) {$(0,q_f)$};
		\node[rstate, right=of q00, xshift=15mm] (q10) {$(1,q_0)$};
		\node[rstate, below=of q10] (q11)  {$(1,q_1)$};
		\node[rstate, below=of q11] (q12) {$(1,q_f)$};
		\node[rstate, right=of q10, xshift=15mm] (q20) {$(2,q_0)$};
		\node[rstate, below=of q20] (q21)  {$(2,q_1)$};
		\node[rstate, below=of q21] (q22) {$(2,q_f)$};
		\node[rstate, right=of q22, xshift=15mm] (q32) {$(3,q_f)$};
		
		\path[->]
		(q00.east) edge node {} (q10)
		(q00.east) edge node {} (q11)
		(q00.east) edge node {} (q12)
		(q01.east) edge node {} (q11)
		(q01.east) edge node {} (q12)
		(q02.east) edge node {} (q12)
		
		(q10.east) edge node {} (q20)
		(q10.east) edge node {} (q21)
		(q10.east) edge node {} (q22)	
		(q11.east) edge node {} (q21)
		(q11.east) edge node {} (q22)
		(q12.east) edge node {} (q22)
		
		(q20.east) edge node {} (q32)
		(q21.east) edge node {} (q32)
		(q22.east) edge node {} (q32)
		;
		
		\end{tikzpicture}
	\end{center}	
	This allows us to obtain the following list of all $\mu\in\spnr{A}(\strs)$:	
	\begin{center}
		\begin{tabular}{|l|c|c|c|}
			\hline 
			ref-word & $(\hat{q}_0,\hat{q}_1,\hat{q}_2,\hat{q}_3)$ & $(\memf_{\hat{q}_0}(x),\ldots,\memf_{\hat{q}_3}(x))$ & $\mu(x)$ \\ 
			\hline 
			$\vop{x} \vcl{x} \ta\ta \ta$& $(q_{f}, q_{f}, q_{f}, q_{f})$  & $(\mc,\mc,\mc,\mc)$ & $\spn{1,1}$ \\ 
			$\vop{x} \ta\vcl{x} \ta \ta$& $(q_{1}, q_{f}, q_{f}, q_{f})$  & $(\mo,\mc,\mc,\mc)$ & $\spn{1,2}$ \\ 
			$\vop{x}  \ta\ta\vcl{x} \ta$& $(q_{1}, q_{1}, q_{f}, q_{f})$ &  $(\mo,\mo,\mc,\mc)$ & $\spn{1,3}$ \\ 
			$\vop{x}  \ta\ta \ta\vcl{x}$& $(q_{1}, q_{1}, q_{1}, q_{f})$ & $(\mo,\mo,\mo,\mc)$ &  $\spn{1,4}$\\ 
			$ \ta\vop{x}\vcl{x} \ta \ta$& $(q_{0}, q_{f}, q_{f}, q_{f})$ &  $(\mw,\mc,\mc,\mc)$& $\spn{2,2}$ \\ 
			$  \ta\vop{x}\ta\vcl{x} \ta$& $(q_{0}, q_{1}, q_{f}, q_{f})$ &  $(\mw,\mo,\mc,\mc)$& $\spn{2,3}$ \\ 
			$  \ta\vop{x}\ta \ta\vcl{x}$& $(q_{0}, q_{1}, q_{1}, q_{f})$ &  $(\mw,\mo,\mo,\mc)$&  $\spn{2,4}$\\ 		
			$  \ta\ta\vop{x}\vcl{x} \ta$& $(q_{0}, q_{0}, q_{f}, q_{f})$ & $(\mw,\mw,\mc,\mc)$ & $\spn{3,3}$ \\ 
			$  \ta\ta \vop{x}\ta\vcl{x}$& $(q_{0}, q_{0}, q_{1}, q_{f})$ &  $(\mw,\mw,\mo,\mc)$&  $\spn{3,4}$\\ 	
			$  \ta\ta \ta\vop{x}\vcl{x}$& $(q_{0}, q_{0}, q_{0}, q_{f})$ & $(\mw,\mw,\mw,\mc)$ & $\spn{4,4}$ \\ 	
			\hline 
		\end{tabular} 
	\end{center}
	\pagebreak[4]
	The NFA $A_G$ looks as follows:
	
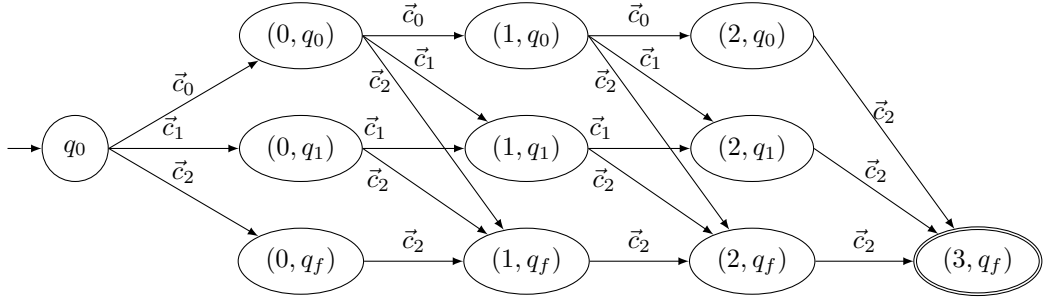
\begin{center}
		\begin{tikzpicture}[on grid, node distance=15mm]
		\node[rstate] (q00) {$(0,q_0)$};
		\node[rstate, below=of q00] (q01)  {$(0,q_1)$};
		\node[rstate, below=of q01] (q02) {$(0,q_f)$};
		\node[state, left=of q01, xshift=-15mm,initial text=,initial by arrow] (q0) {$q_0$};
		\node[rstate, right=of q00, xshift=15mm] (q10) {$(1,q_0)$};
		\node[rstate, below=of q10] (q11)  {$(1,q_1)$};
		\node[rstate, below=of q11] (q12) {$(1,q_f)$};
		\node[rstate, right=of q10, xshift=15mm] (q20) {$(2,q_0)$};
		\node[rstate, below=of q20] (q21)  {$(2,q_1)$};
		\node[rstate, below=of q21] (q22) {$(2,q_f)$};
		\node[rstate, right=of q22, xshift=15mm,accepting] (q32) {$(3,q_f)$};
		
		\path[->]
		(q0.east) edge[above] node {$\memf_{0}$} (q00)
		(q0.east) edge[above] node {$\memf_{1}$} (q01)
		(q0.east) edge[above] node {$\memf_{2}$} (q02) 	 	
		(q00.east) edge[above] node {$\memf_{0}$} (q10)
		(q00.east) edge[above] node {$\memf_{1}$} (q11)
		(q00.east) edge[below, very near start] node {$\memf_{2}$} (q12)
		(q01.east) edge[above, very near start] node {$\memf_{1}$} (q11)
		(q01.east) edge[below, very near start] node {$\memf_{2}$} (q12)
		(q02.east) edge[above] node {$\memf_{2}$} (q12)
		
		(q10.east) edge[above] node {$\memf_{0}$} (q20)
		(q10.east) edge[above] node {$\memf_{1}$} (q21)
		(q10.east) edge[below, very near start] node {$\memf_{2}$} (q22)	
		(q11.east) edge[above, very near start] node {$\memf_{1}$} (q21)
		(q11.east) edge[below, very near start] node {$\memf_{2}$} (q22)
		(q12.east) edge[above] node {$\memf_{2}$} (q22)
		
		(q20.east) edge[above] node {$\memf_{2}$} (q32)
		(q21.east) edge[above] node {$\memf_{2}$} (q32)
		(q22.east) edge[above] node {$\memf_{2}$} (q32)
		;
		
		\end{tikzpicture}
	\end{center}
	As $\memf_{i} \neq \memf_{j}$ if $i\neq j$, this NFA could actually be interpreted as a DFA. As a result, the running time of the enumeration algorithm is lower than in the worst case.
\end{example}
\begin{example}
	In order to examine a case where constructing $A_D$ is necessary, consider the following functional vset-automaton~$A$:
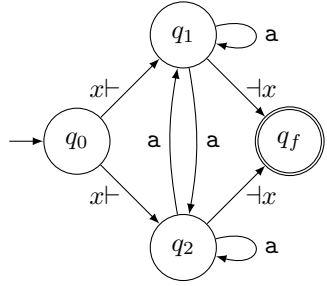
	\begin{center}
		\begin{tikzpicture}[on grid, node distance =20mm,every loop/.style={shorten >=0pt}]
		\node[state,initial text=,initial by arrow] (q0) {$q_0$};
		\node[state,above right=of q0] (q1) {$q_1$};
		\node[state,below right=of q0] (q2) {$q_2$};
		\node[state,above right=of q2,accepting] (qf) {$q_f$};	
		\path[->]
		(q0) edge[left] node {$\vop{x}$} (q1)
		(q0) edge[left] node {$\vop{x}$} (q2)
		(q1) edge[right] node {$\vcl{x}$} (qf)
		(q2) edge[right] node {$\vcl{x}$} (qf)	
		(q1) edge[bend left=10,right] node {$\ta$} (q2)
		(q2) edge[bend left=10,left] node {$\ta$} (q1)
		(q1) edge[loop right] node {$\ta$} (q1)
		(q2) edge[loop right] node {$\ta$} (q2)
		;
		\end{tikzpicture}
	\end{center}
	It is easily seen that $A$ is functional, and the variable configurations are defined by $\memf_{q_0}(x)=\mw$, $\memf_{q_1}(x)=\memf_{q_2}(x)=\mo$, and $\memf_{q_f}(x)=\mc$. Furthermore, for every $\strs\in\ta^*$, $\spnr{A}(\strs)$ contains only a single $\mu$, with $\mu(x)=\spn{1,|\strs|+1}$, which corresponds to the ref-word $\vop{x} \strs \vcl{x}$. Nonetheless, for each such $\strs$, there are $2^{|\strs|}$ paths from $q_0$ to $q_f$ in $A$. Now consider the case of $\strs=\ta^3$. The corresponding graph $G$ looks as follows:
	
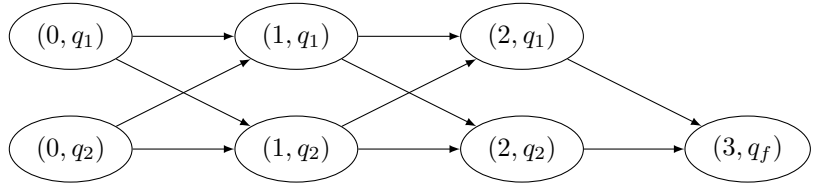
\begin{center}
		\begin{tikzpicture}[on grid, node distance=15mm]
		\node[rstate] (n01) {$(0,q_1)$};
		\node[rstate, below=of n01] (n02) {$(0,q_2)$};
		
		\node[rstate, right=of n01, xshift=15mm] (n11) {$(1,q_1)$};
		\node[rstate, below=of n11] (n12) {$(1,q_2)$};
		\node[rstate, right=of n11,xshift=15mm] (n21) {$(2,q_1)$};
		\node[rstate, below=of n21] (n22) {$(2,q_2)$};
		
		\node[rstate, right= of n22,xshift=15mm] (n3f) {$(3,q_f)$};
		
		\path[->]
		(n01) edge[] node {} (n11)
		(n01) edge[] node {} (n12)
		(n02) edge[] node {} (n11)
		(n02) edge[] node {} (n12)
		
		(n11) edge[] node {} (n21)
		(n11) edge[] node {} (n22)
		(n12) edge[] node {} (n21)
		(n12) edge[] node {} (n22)
		
		(n21) edge[] node {} (n3f)
		(n22) edge[] node {} (n3f)
		;
		\end{tikzpicture}
	\end{center}
	By adding a new starting state $q_0$, we construct the NFA $A_G$:
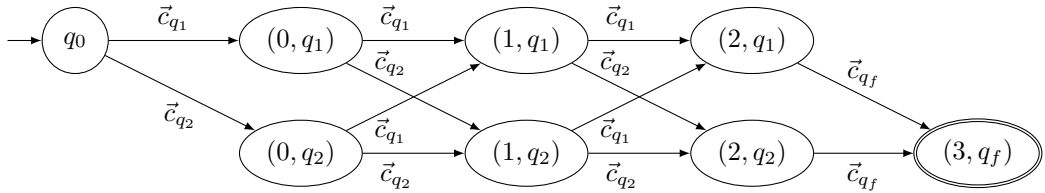
	\begin{center}
		\begin{tikzpicture}[on grid, node distance=15mm]
		\node[rstate] (n01) {$(0,q_1)$};
		\node[state,initial text=,initial by arrow, left=of n01,xshift=-15mm] (q0) {$q_0$};
		\node[rstate, below=of n01] (n02) {$(0,q_2)$};
		
		\node[rstate, right=of n01, xshift=15mm] (n11) {$(1,q_1)$};
		\node[rstate, below=of n11] (n12) {$(1,q_2)$};
		\node[rstate, right=of n11,xshift=15mm] (n21) {$(2,q_1)$};
		\node[rstate, below=of n21] (n22) {$(2,q_2)$};
		
		\node[rstate, right= of n22,xshift=15mm,accepting] (n3f) {$(3,q_f)$};
		
		\path[->]
		(q0) edge[above] node {$\memf_{q_1}$} (n01)
		(q0) edge[below] node {$\memf_{q_2}$} (n02)	
		
		(n01) edge[pos=0.33,above] node {$\memf_{q_1}$} (n11)
		(n01) edge[pos=0.33,above] node {$\memf_{q_2}$} (n12)
		(n02) edge[pos=0.33,below] node {$\memf_{q_1}$} (n11)
		(n02) edge[pos=0.33,below] node {$\memf_{q_2}$} (n12)
		
		(n11) edge[pos=0.33,above] node {$\memf_{q_1}$} (n21)
		(n11) edge[pos=0.33,above] node {$\memf_{q_2}$} (n22)
		(n12) edge[pos=0.33,below] node {$\memf_{q_1}$} (n21)
		(n12) edge[pos=0.33,below] node {$\memf_{q_2}$} (n22)
		
		(n21) edge[above] node {$\memf_{q_f}$} (n3f)
		(n22) edge[below] node {$\memf_{q_f}$} (n3f)
		;
		\end{tikzpicture}
	\end{center}
	Now, note that $A_G$ is not deterministic, as $\memf_{q_1}=\memf_{q_2}$. In fact, $\lang(A_G)$ consists only of the word $(\memf_{q_1})^3 \memf_{q_f}$, which corresponds to $\mu$ with $\mu(x)=\spn{1,4}$, the only element of $\spnr{A}(\strs)$.
\end{example}

\subsection{Proof of Lemma~\ref{lem:regexToAutomaton}}
\begin{replemma}{\ref{lem:regexToAutomaton}}
	\lemmaregexToAutomaton
\end{replemma}
\begin{proof}
	Let $\alpha$ be a regex formula, and define $V\df
	\svars(\alpha)$. Assume that $\alpha$ is represented as its syntax
	tree. We first rewrite $\alpha$ into a regular expression
	$\hat{\alpha}$ over the alphabet $\alphabet\cup\xalphabet_{V}$ with
	$\lang(\hat{\alpha})=\rlang(\alpha)$. This is done by recursively
	replacing every node that represents a variable binding
	$\bind{x}{\beta}$ with the concatenation $\vop{x}\cdot \beta\cdot
	\vcl{x}$. Then the length of $\hat{\alpha}$ is linear in the length
	of $\alpha$; and the rewriting is possible in linear time as well.
	
	Next, we convert $\hat{\alpha}$ into an $\emptyword$-NFA $A$ with
	$\lang(A)=\lang(\hat{\alpha})$. Using the Thompson construction
	(cf.~e.\,g.\ \cite{hop:int}), this is possible in linear
	time. Furthermore, both the number of states and the number of
	transitions of $A$ are linear in the length of $\hat{\alpha}$ (and,
	hence, also in the length of $\alpha$). Finally, note that the
	construction ensures that $A$ has only one accepting state (recall
	that vset-automata are required to have a single accepting state).
	
	This allows us to interpret $A$ as a vset-automaton with variable
	set $V$, using $\rlang(A)=\lang(A)$. As
	$\lang(A)=\lang(\hat{\alpha})=\rlang(\alpha)$ holds by definition,
	we know that $\rlang(A)=\rlang(\alpha)$. Furthermore, as $\alpha$ is
	functional, $\rlang(\alpha)=\refl(\alpha)$ holds, which implies
	$\refl(A)=\rlang(A)=\refl(\alpha)$. This allows us to make two
	conclusions: Firstly, that $\refl(A)=\rlang(A)$, hence $A$ is
	functional. And, secondly, that $\spnr{A}=\spnr{\alpha}$, as
	$\refl(A)=\refl(\alpha)$ implies
	$\refl(A,\strs)=\refl(\alpha,\strs)$ for all $\strs\in\alphabet^*$.
\end{proof}

\subsection{Proof of Proposition~\ref{prop:keyProperty}}

\begin{repproposition}{\ref{prop:keyProperty}}
	\propkeyProperty
\end{repproposition}

\begin{proof}
	\newcommand{\Akey}{A_{x}}
	The proof uses a modification of the intersection construction for NFAs. Given a functional vset-automaton $A=(V,Q,q_0,q_f,\delta)$, our goal is to construct an NFA $\Akey$ such that $\lang(\Akey)$ is empty if and only if $x$ is not a key attribute. More specifically, the language consists of all $\strs$ for which there exist $\mu,\mu'\in \spnr{A}(\strs)$ such that $\mu(x)=\mu'(x)$ and $\mu(y)\neq \mu'(y)$ for some $y\in V$. Without loss of generality, we can assume that $A$ has at least two variables, and that we consider only non-empty strings (as $\spnr{A}(\emptyword)$ contains at most one element). For our complexity analysis, let $n\df |Q|$ and $v\df |V|$.
	
	Before we discuss the actual construction, we bring $A$ into a more convenient form. As in the proof of Theorem~\ref{thm:delayAlgorithm}, we use the fact that every state of $A$ has a uniquely defined variable configuration $\memf_q\colon X\to \memstat$ (also see Section~\ref{sec:delay}). Hence, the first step of the construction is computing the variable configurations of $A$. Recall that, if $A$ has $n$ states and $m$ transitions, this can be done in time $O(m+n)$. We can also assume that every state is reachable from $q_0$, and that $q_f$ is reachable from each state. We now replace all variable transitions with $\emptyword$-transitions, and obtain a new transition function $\delta'$ from $\delta$ by removing all $\emptyword$-transitions. Using the standard procedure, this takes time $O(n^3)$ (in contrast to previous proofs, we do not need to use  the tighter bound $O(mn)$).  
	
	The main idea of the construction is that $\Akey$ is an NFA over $\alphabet$ that simulates two copies of $A$ in parallel. Both copies have the same behavior on $x$, but different behavior on (at least) one witness variable $y$. In addition to this, $\Akey$ uses its states to keep track whether such a $y$ has been found. Hence, we define $\Akey\df (Q_x,q_{0,x},q_{f,x},\delta_x)$, where
	\begin{align*}
		Q_x \df&\: \{(0,q_1,q_2)\mid q_1,q_2\in Q, \memf_{q_1}=\memf_{q_2}\}\\
		\cup&\:  \{(1, q_1,q_2) \mid  q_1,q_2\in Q, \memf_{q_1}(x)=\memf_{q_2}(x)\}.
	\end{align*}
	Intuitively, each state triple contains the state of each of the two copies of $A$, and a bit that encodes with 1 or 0 whether a witness $y$ has been found  or not (respectively). Hence, we require for all state pairs that their variable configurations on $x$ are identical; but as long as no witness $y$ has been found, the variable configurations of all other variables  also have to be identical. Following this intuition, we define $q_{0,x}\df (0,q_0,q_0)$ and $(1,q_f,q_f)$. 
	Finally, we construct $\delta_x$ is  as follows: For each pair of transitions $q_1\in\delta'(p_1,\sigma)$ and $q_2\in\delta'(p_2,\sigma)$ that satisfies $\memf_{q_1}(x)=\memf_{q_2}(x)$, we first define 
	$(1,q_1,q_2)\in \delta_x((1,p_1,p_2),\sigma)$ (as the bit signifies that a witness $y$ has been found, we simply continue simulating the two copies of $A$).
	In addition to this,  if $\memf_{q_1}=\memf_{q_2}$, we also define 
	\begin{itemize}
		\item $(0,q_1,q_2)\in \delta_x((0,p_1,p_2),\sigma)$ if $\memf_{q_1}=\memf_{q_2}$, or
		\item $(1,q_1,q_2)\in \delta_x((0,p_1,p_2),\sigma)$ if $\memf_{q_1}\neq \memf_{q_2}$, i.\,e., there is a variable $y$ with $\memf_{q_1}(y)\neq \memf_{q_2}(y)$.		
	\end{itemize}
	In the first case, all variables have identical behavior, so we have not found a witness $y$, and do not change the bit. In the second case, there is a witness $y$, which allows us to change the bit. Note that we can precompute the sets of all $(q_1,q_2)\in Q$ with  $\memf_{q_1}=\memf_{q_2}$ or $\memf_{q_1}\neq \memf_{q_2}$ in time $O(vn^2)$.
	
	Now $\Akey$ describes exactly the language of strings $\strs$ for which there exist demonstrate that $x$ does not have the key property.  As $A'$ has $O(n^2)$ transitions, we can construct $\Akey$ in time $O(vn^2 + n^4)=O(n^4)$, and emptiness of $\lang(\Akey)$ can also be decided in time $O(n^4)$, by checking whether $\Akey$ contains a path from $(0,q_0,q_0)$ to $(1,q_f,q_f)$. Note that such a path also provides a witness input string $\strs$, and, by decoding the variable configurations of the states along the path, the corresponding $\mu,\mu'\in\spnr{A}(\strs)$ with $\mu(x)=\mu'(x)$ and $\mu\neq \mu'$. As a side-effect, this construction also shows that the shortest witness that $x$ does not have the key property is of length $O(n^2)$.
\end{proof}
\subsection{Proof of Lemma~\ref{lem:projection}}
\begin{replemma}{\ref{lem:projection}}
	\lemprojection
\end{replemma}
\begin{proof}
	Let $A$ be a functional vset-automaton, $V\df\svars(A)$, and
	$Y\subseteq V$. We define the morphism $h_Y\colon (\xalphabet_V\cup
	\alphabet)^* \to (\xalphabet_Y\cup\alphabet)^*$ for every $g\in
	(\xalphabet_V\cup\alphabet)$ by $h_Y(g)=\emptyword$ if
	$g\in(\xalphabet_V\setminus \xalphabet_Y)$, and $h_Y(g)\df g$ for all
	$g\in(\xalphabet_Y\cup\alphabet)$. In other words, $h_Y$ erases all
	$\vop{x}$ and $\vcl{x}$ with $x\notin Y$, and leaves all other
	symbols unchanged.
	
	We obtain $A_Y$ from $A$ by replacing the label $g$ of each
	transition with $h(g)$. In other words, for each $x\notin Y$, each
	transition with $\vop{x}$ or $\vcl{x}$ is replaced with an
	$\emptyword$-transition. Clearly, this is possible in linear time,
	and $\rlang(A_Y)= h_Y(\rlang(A))$. As $A$ is functional, this
	implies $\rlang(A_Y)=h_Y(\refl(A))$.
	
	Now assume that $A_Y$ is not functional, i.\,e., that there
	exists an $\strr\in (\refl(A_Y)\setminus \rlang(A_Y))$. Then $\strr$
	is not valid, which means that there is an $y\in Y$ such that
	$\vop{y}$ or $\vcl{y}$ does not occur exactly once, or both
	symbols occur in the wrong order. By definition of $A_Y$,
	there is an $\hat{\strr}\in\rlang(A)$ with $h_Y(\hat{\strr})=\strr$; and
	as $h_Y$ erases only symbols of $\xalphabet_V\setminus \xalphabet_Y$
	and leaves all other symbols unchanged, this means that
	$\hat{\strr}$ is also not valid, which contradicts our assumption
	that $A$ is functional.  Hence,
	$\refl(A_Y)=\rlang(A_Y)=h_Y(\refl(A))$ holds. This also
	implies $\spnr{A_Y}=\spnr{\pi_Y(A)}$.
\end{proof}

\subsection{Proof of Lemma~\ref{lem:union}}
\begin{replemma}
	{	\ref{lem:union}}
	\lemunion
\end{replemma}
\begin{proof}
	For $1\leq i \leq k$, let $n_i$ denote the number of states and
	$m_i$ denote the number of transitions of $A_i$. We can construct
	$A$ by using the standard construction for union: We add a new
	initial and a new final state, and connect these with
	$\emptyword$-transitions to the ``old'' inital and final states of
	the $A_i$. Clearly, $\rlang(A)= \bigcup_{i=1}^k \rlang(A_i)$, and as
	each $A_i$ is functional, $\refl(A)=\bigcup_{i=1}^k\refl(A_i)$
	follows. Hence, $A$ is functional, and $\spnr{A}=\bigcup_{i=1}^k
	\spnr{A_i} = \spnr{A_1\cup\cdots \cup A_k}$.  Adding the new states
	takes time $O(k)$, and outputting $A$ takes time $O(\sum_{i=1}^k
	(n_i + m_i))$. As this is the same size as the input of the
	algorithm, the construction runs in linear time.
	
	Now, let $\strs\in\alphabet^*$, and define $\ell\df|\strs|$. If we now
	construct the graph $G$ for $A$ and $\strs$ as in the proof of
	Theorem~\ref{thm:delayAlgorithm}, each $A_i$ leads to its own
	subgraph of $G$; and for distinct $A_i$, $A_j$, these subgraphs are
	disjoint. Hence, $G$ has $O(\ell \sum_{i=1}^k n_i)$ nodes, and as
	each node of the subgraph for $A_i$ can have at most $n_i$
	successors, $G$ has $O(\ell \sum_{i=1}^k n^2_i)$ edges. This allows
	us to compute $A_G$ in $O(\sum_{i=1}^k (\ell n^2_i + m_i n_i))$.
	
	Following the same reasoning, we can conclude that for $A_G$,
	$\minword$ runs in time $O(\ell \sum_{i=1}^k n^2_i)$, as in each of
	the $O(\ell)$ iterations of the for-loop, $S_i$ contains
	$O(\sum_{i=1}^k n_i)$ states, and each state that was derived from
	$A_i$ has $O(n_i)$ successors. As this determines the complexity of
	$\enumalgo$, we conclude that $\spnr{A}(\strs)$ can be enumerated
	with delay $O(\ell \sum_{i=1}^k n^2_i)$ after $O(\sum_{i=1}^k (\ell
	n^2_i + m_i n_i))$ preprocessing. As $n_i\leq n$ and $m_i\leq m$ for
	all $1\leq i \leq k$, this can be simplified to a delay of $O(k\ell
	n^2)$ after a preprocessing of $O(k\ell n^2+ kmn)$.
\end{proof}

\subsection{Proof of Lemma~\ref{lem:joinVset}}
\begin{replemma}{\ref{lem:joinVset}}
	\lemjoinVset
\end{replemma}
\begin{proof}
	\newcommand{\tclos}{\mathcal{T}}
	\newcommand{\veclos}{\mathcal{VE}}
	\newcommand{\Astrict}{A_{\mathsf{strict}}}
	Assume we are given two functional vset-automata $A_1$ and $A_2$, with  $A_i=(V_i,Q_i,q_{0,i},q_{f,i},\delta_i)$.  Let $v_i\df|V_i|$, $n_i \df |Q_i|$, and let $m_i$ denote the number of transitions of $A_i$.  Furthermore, let $V \df V_1 \cup V_2$ and define $v\df|V|$. We assume that in each $A_i$, all states are reachable from $q_{0,i}$, and $q_{f,i}$ can be reached from every state.
	In order to construct a functional vset-automaton $A$ with $\spnr{A}=\spnr{A_1\join A_2}$, we modify the product construction for the intersection of two NFAs. 
	
	Note that construction has to deal with the fact that the ref-words from $\rlang(A_1)$ and $\rlang(A_2)$ can use the variables in different orders. As a simple example, consider $\strr_1 \df \vop{x} \vop{y} \ta \vcl{y}\vcl{x}$ and $\strr_2 \df \vop{y} \vop{x} \ta \vcl{x}\vcl{y}$. Then $\mu^{\strr_1}=\mu^{\strr_2}$, although $\strr_1\neq \strr_2$. Hence, we cannot directly  apply the standard construction for NFA-intersection; and for complexity reasons, we do not rewrite the automata to impose an ordering on successive variable operations. Instead, we use the variable configurations of $A_1$ and $A_2$.
	
	\partitle{Construction:} In order to simplify the construction, we slightly abuse the definition, and label the variable transitions of $A$ with sets of variable transitions, instead of single transitions. We discuss the technical aspects of this  of this at the end of the proof (in particular its effect on the complexity). This decision allows us to restrict the number of states in $A$, which in turn simplifies the construction: 
	Like the standard construction of the intersection, the main idea of the proof is that  $A$ simulates each of $A_1$ and $A_2$ in parallel. Hence, its state set is a  subset of $Q_1\times Q_2$. But as we shall see, this construction uses the variable configurations $\Memf_i$ of $A_i$ to determine how $A$ should act on the variables (instead of the variable transitions). In particular, we shall require that all states of $A$ are consistent, where $(q_1,q_2)\in Q_1\times Q_2$ is \emph{consistent} if $\Memf_1(q_1)(x)=\Memf_2(q_2)(x)$ for all $x\in V_1\cap V_2$. In other words, the variable configurations of $q_1$ and $q_2$ agree on the common variables of $A_1$ and $A_2$. 
	
	By using sets of variable operations, we can disregard the order in which $A_1$ and $A_2$ process the common variables; and as variable transitions connect only consistent states of $A$ that encode consistent pairs of states, the fact that $A_1$ and $A_2$ are functional shall ensure that $A$ is also functional. 
	
	Before the algorithm constructs $A$, it computes the following sets and functions:
	\begin{enumerate}
		\item the variable configurations function $\Memf_i$ for $A_i$,
		\item the set $Q$ of all consistent $(q_1,q_2)\in Q_1\times Q_2$, 
		\item the sets $Q^{\equiv}_i$ of all $(p,q)\in Q_i\times Q_i$ with $\Memf_i(p)=\Memf_i(q)$,
		\item the $\emptyword$-closures $\eclos_i\colon Q_i \to 2^{Q_i}$, where for each $p\in Q_i$, $\eclos_i(p)$ contains all $q\in Q_i$ that can be reached from $p$ by using only $\emptyword$-transitions,
		\item the \emph{variable-$\emptyword$-closures} $\veclos_i\colon Q_i \to 2^{Q_i}$, where for each $p\in Q_i$, $\veclos_i(p)$ contains all $q\in Q_i$ that can be reached from $p$ by using only transitions from  $\{\emptyword\}\cup\xalphabet_{V_i}$, 
		\item for each $\sigma\in\alphabet$, a function $\tclos^{\sigma}_i\colon Q_i \to 2^{Q_i}$, that is defined by  $\tclos^{\sigma}_i(p)\df \bigcup_{q\in\delta_i(p,\sigma)} \eclos_i(q)$ for each $p\in Q_i$.
	\end{enumerate}
	In other words,  $\veclos_i$ can be understood as replacing all variable transitions in $A_i$ with $\emptyword$-transitions, and then computing the $\emptyword$-closure; while  $\tclos^{\sigma}_i(q)$ contains all states that can be reached by processing exactly one terminal transition, and then arbitrarily many $\emptyword$-transitions. 
	
	As mentioned in the proof of Theorem~\ref{thm:delayAlgorithm}, the variable configurations $\Memf_i$ can be computed in $O(v_im_i)$. These can then be used to compute $Q$ in time $O(|V_1\cap V_2| n_1 n_2)$, and each $Q^{\equiv}_i$ in time $O(v_i n_i^2)$. 
	By using the standard algorithm for transitive closures in directed graphs (see e.\,g.\ Skiena~\cite{ski:alg}), each $\eclos_i$ and $\veclos_i$ can be computed in time $O(n_i m_i)$. 
	Finally, each of the $\tclos^{\sigma}_{i}$ can the be computed in time $O(m n_i)$, by processing each transition of $A_i$ and then using $\eclos_i$.  
	
	Now,  $O(v_i n_i^2)$ dominates $O(v_i m_i)$. Furthermore, for at least one $i$, $O(v_i n_i^2)$ also dominates $O(|V_1\cap V_2| n_1 n_2)$, as $|V_1\cap V_2|\leq \min\{v_1, v_2\}$. 
	Hence, the total running time of this precomputations adds up to $O(v_1 n_1^2 + m_1 n_1 + v_2 n_2^2 + m_2 n_2)$. 
	
	We are now ready to define $A\df (V,Q,q_0,q_f,\delta)$, where  $Q$ is the set we defined above, $q_0\df (q_{0,1}, q_{0,2})$, and $q_f \df (q_{f,1},q_{f,2})$. Before we define $\delta$, note that $q_0$ and $q_f$ are always consistent---as $A_1$ and $A_2$ are functional, the variable configurations of the initial and final states map all variables to $\mo$ and $\mc$, respectively.  For the same reason, for all $p\in Q_i$, all $\sigma\in\alphabet$,  $\Memf_i(q)=\Memf_i(p)$ must hold for  all $q\in \tclos^{\sigma}_i(p)$ or $q\in \eclos_i(p)$.
	
	We now define $\delta$ as follows:
	\begin{enumerate}
		\item For every pair of states $(q_1,q_2)$ with $q_i\in\eclos_i(q_{0,i})$, we add an $\emptyword$-transition from $(q_{0,1},q_{0,2})$ to $(q_1,q_2)$. 
		\item For every pair of states $(p_1,p_2)\in Q$ and every $\sigma\in \alphabet$, we compute all $(q_1, q_2) \in \tclos_1^{\sigma}(p_1)\times \tclos^{\sigma}_2(p_2)$. To each such $(q_1,q_2)$, we add a transition with label $\sigma$ from $(p_1,p_2)$. 
		\item For every pair of states $(p_1,p_2)\in Q$, we compute all $(q_1,q_2)\in \veclos_1(p_1)\times \veclos_2(p_2)$. If $(q_1,q_2)\in Q$, and $\Memf_i(q_i)\neq \Memf_i(p_i)$, we add a variable transition from $(p_1,p_2)$ to $(q_1,q_2)$ that contains exactly the operations that map each $\Memf_i(p_i)$ to $\Memf_i(q_i)$. 
	\end{enumerate}
	Each of these three rules guarantees that the destination $(q_1,q_2)$ of each transition is indeed a consistent pair of states (the third rule requires this explicitly, while the other two use the properties of $\eclos_i$ and $\tclos^{\sigma}_i$ mentioned above together with the fact that $(p_1,p_2)$ is consistent). 
	
	\partitle{Correctness and complexity:} To see why this construction is correct, consider the following: The basic idea is that $A$ simulates $A_1$ and $A_2$ in parallel. The actual work of the simulation is performed by the transitions that were derived from the second or the third rule. The former simulate the behavior of $A_1$ and $A_2$ on terminal letters, while the letter simulate sequences of variable actions. In particular, as we use the variable-$\emptyword$-closures, a transition of $A$ can simulate all sequences of variable actions that lead to a consistent state pair. Hence, it does not matter in which order $A_1$ or $A_2$ would process variables, the fact that the state of $A$ is a consistent pair means that both automata agree on their common variables.
	
	In order to keep the second and third rule simple, we also include the first rule: The effect of these transitions is that $A$ can simulate that $A_1$ or $A_2$ changes states via $\emptyword$-transitions before it starts processing terminals or variables. Finally, to see that $A$ is functional, assume that it is not. Then there is a ref-word $\strr\in\rlang(A)\setminus \refl(A)$ that is not valid. Let $\strr_i$ be the result of projecting $\strr$ to $\alphabet\cup \xalphabet_{V_i}$ (i.\,e., removing all elements of $\xalphabet_{V_j}$ with $j\neq i$). Then at least one of $\strr_1$ or $\strr_2$ is also not valid; we assume $\strr_1$. But according to the  definition of $A$, this means that $\strr_1\in\rlang(A_1)$, which implies $\rlang(A_1)\neq\refl(A_1)$, and contradicts our assumption that $A_1$ is functional. Hence, $A$ must be functional.
	
	Regarding the complexity of the construction, note that $A$ has $O(n_1 n_2)$ states and $O(n_1^2 n_2^2)$ transitions, and can be constructed in time $O(n_1^2 n_2^2)$. All transitions that follow from the first rule can obviously be computed in $O(n_1 n_2)$.
	For the second rule, for each of $O(n_1n_2)$ many pairs $(p_1,p_2)\in Q$, we enumerate $O(n_1 n_2)$ many pairs $(q_1, q_2) \in \tclos(p_1)_1^{\sigma}\times \tclos(p_2)^{\sigma}_2$, and check whether $(q_1,q_2)\in Q$. As we precomputed $Q$, this check is possible in $O(1)$, which means that these transitions can be constructed in time $O(n_1^2 n_2^2)$. Likewise, for the third rule, we use that we precomputed the sets $Q^{\equiv}_i$, which allows us to check $\Memf_i(p)\neq \Memf_i(q)$ in time $O(1)$ as well.
	
	For the third rule, we need to take into consideration that the sets of variable operations have to be computed. Doing this na\"{i}vely would bring the complexity to $O((v_1+v_2)n_1^2 n_2^2)$, but using the right representation, these sets can be precomputed independently for each $A_i$ in time $O(v_im_i)$, which allows us to leave the time for the variable transitions unchanged. 
	
	Up to this point, the complexity of the construction with precomputations is $O(v_1 n_1^2 + m_1 n_1 + v_2 n_2^2 + m_2 n_2 + n_1^2 n_2^2)$. For $v,m,n$ with $v_i\in O(v)$, $m_i\in O(m)$, $n\in O(n)$, this becomes 
	$O(v n^2 + m n  + n^4)$, which we can simplify to $O(n^4)$, as $v\leq n$, and $m\leq n^2$. Note that this is the same complexity as constructing the product automaton for the intersection of two NFAs.
	
	If we want to strictly adhere to the definition of vset-automata, we can rewrite $A$ into an equivalent vset-automaton $\Astrict$, by taking each variable transition that is labeled with a set $S$, and replacing it with a sequence of states and transitions that enumerates the elements of $S$. This would add $O((v_1+v_2)n_1^2 n_2^2)$ states and the same number of transitions, and increase the complexity of the construction by $O((v_1+v_2)n_1^2 n_2^2)$, netting a total complexity of $O(m_1 n_1 + m_2 n_2 + (v_1+v_2)n_1^2 n_2^2)$. Under the assumptions from the last paragraph, this can be simplified to $O(v n^4)$.
	
	While a more detailed analysis or a more refined construction might avoid this increase, it might be more advantageous to generalize the definition of vset-automata to allow sets of variable operations on transitions, as the complexities of computing the variable configurations and of Theorem~\ref{thm:delayAlgorithm} generalize to this model (as long as the number of variables is linear in the number of states).
	
	In particular, if our actual goal is evaluating $A_1\join A_2$ by using Theorem~\ref{thm:delayAlgorithm}, we can skip the step of rewriting $A$ into $\Astrict$, and directly construct $G$ from $A$ (or, without explicitly constructing $A$, from the $\eclos_i$ and $\tclos^{\sigma}_i$ functions). The total time for preprocessing is then $O(v_1 n_1^2 + m_1 n_1 + v_2 n_2^2 + m_2 n_2 + \ell n_1^2 n_2^2)$, or, if simplified, $O(\ell n^4)$. This is also the complexity of the delay.
	
	Naturally, all these constructions can be extended to the join of $k$ vset-automata, where the complexities become $O(n^{2k})$ for the construction of the vset-automaton $A$ with sets of variable operations, as well as $O(vn^{2k})$ for the construction of $\Astrict$ or $O(\ell n^{2k})$ for the preprocessing and the delay of the enumeration algorithm.
\end{proof}

\subsection{Proof of Theorem~\ref{thm:equalitiesNPW1hard}}
\begin{reptheorem}{\ref{thm:equalitiesNPW1hard}}
	\thmequalitiesNPWonehard
\end{reptheorem}
\begin{proof}
	We prove the claim by modifying the proof of Theorem~\ref{thm:aacyclicW1}:
	Let $\strs$ and $\gamma$ be defined as in that proof. 
	However, we do not use regex formulas  $\delta_l$ to ensure that all variables $y_{i,l}$ and $x_{l,j}$ with $1\leq i < l < j \leq k$  map to the same substring $\str{v}_l$. Instead, for each $1\leq l \leq k$,  we  define a sequence $S_l$ of $k-2$ binary string equality selections that is equivalent to the $(k-1)$-ary string equality selection on the variables $y_{1,l}$ to $y_{l-1,l}$ and $x_{l,l+1}$ to $x_{l,k}$. 
	We then define our query $q$ as 
	$$
	q\df \pi_{\emptyset}
	S_{1} \cdots S_{k} \gamma.
	$$
	Note that 
	being able to use string equality predicates, we do not need to iterate through all of the possible $\str{v}$ as in the $\delta_l$. 
	Therefore the proof of correctness of the reduction used in the proof of Theorem~\ref{thm:aacyclicW1} can be simply adapted to show correctness of this reduction.
	Observe that this is an FPT reduction with parameter $|q|$ since $|q|$ is of size $O(k^2)$  (i.e., the number of string equality predicates we use depends only on the clique size). 
	Finally, we remark that like the proof of Theorem~\ref{thm:aacyclicW1}, this proof can be adapted to a binary alphabet by using the standard coding techniques.
\end{proof}

\subsection{Proof of Theorem~\ref{thm:enumKEqualitiesPolyDelay}}
\begin{reptheorem}
	{\ref{thm:enumKEqualitiesPolyDelay}}
	\thmenumKEqualitiesPolyDelay
\end{reptheorem}

\begin{proof}
	\newcommand{\qeq}{q_{\mathsf{eq}}}
	Let $A$ be a functional vset-automaton with $n$ states and $m$ transitions, $\str s$ an input string, $k$ a fixed integer, and $x_i,y_i \subseteq \svars(A)$ for $1\leq i \leq k$. Let $\ell\df |\strs|$, and let $v\df \svars(A)$.
	We describe an algorithm for enumerating the tuples in 
	$\spnr{\zeta^=_{x_1,y_1}\cdots \zeta^=_{x_k,y_k} A}(\str s)$ in polynomial delay.
	
	
	\partitle{Algorithm:} Iterating over all of the possible couples of spans of $\str s$, we can obtain all pairs $s_1, s_2$ of spans of $\strs$
	that refer to equal substrings; i.\,e., $\strs_{s_1}=\strs_{s_2}$. This can be done na\"ively in  time $O(\ell^4)$. Using these pairs, we construct a functional vset-automata $\Aeq $ with $\svars(A^=)=\{x_i,y_i\mid 1\leq i \leq k\}$ such that $\mu\in\spnr{\Aeq }$ if and only if $\strs_{\mu(x_i)}=\strs_{\mu(y_i)}$ holds for all $1\leq i \leq k$. In other words, $\Aeq $ describes exactly those spans of $\strs$ that satisfy the string equality selections. 
	
	We ensure that $\Aeq $ has $O(\ell^{3k+1})$ states, and each state except the initial state and the final state has exactly one outgoing transition. More specifically, $\Aeq $ encodes each possible $\mu\in\spnr{A}(\strs)$ in a path of states. Each of these $\mu$ is characterized by selecting for each $1\leq i \leq k$ the length of $\strs_{\mu(x_i)}$ (and, hence, $\strs_{\mu(y_i)}$), as well as the starting positions of $\mu(x_i)$ and $\mu_(y_i)$. Hence, there are $O(\ell^{3k})$ possible $\mu$, and each can be encoded in a path of $O(\ell+k)$ successive states if we take variable operations into account. If we allow multiple variable actions on a single transition (see the Proof of Lemma~\ref{lem:joinVset}), this is possible in $O(\ell)$ states; but even if we do not, we can safely assume that $k\leq \ell$, which also gives $O(\ell)$ states.
	
	The initial state of $\Aeq $ non-deterministically choses with an $\emptyword$-transition which of these paths to take, which means that $\Aeq $ can be constructed with $O(\ell^{3k+1})$ states and transitions. Moreover, we can build $\Aeq $ while enumerating the pairs $s_1,s_2$ as described above, which means that the construction takes time $O(\ell^{3k+1})$. Note that if two selections have a common variable, we can lower the exponent by two, as we only need to account once for the length and the starting position.  For example, $\Aeq $ for $\sel_{x,y}\sel_{y,z} A$
	only needs $O(\ell^5)$ states, instead of $O(\ell^7)$. 
	
	We can now use Lemma~\ref{lem:joinVset} to construct a functional vset-automaton $\Ajoin$ with  $\spnr{\Ajoin}=\spnr{A \join \Aeq }$. As $\spnr{\Ajoin}(\strs)=\spnr{\sel_{x_1,y_1}\cdots\sel_{x_k,y_k}A}(\strs)$ holds, this solves our problem. 
	
	\partitle{Complexity analysis:} If we directly use Lemma~\ref{lem:joinVset} together with Theorem~\ref{thm:delayAlgorithm} na\"ively, preprocessing and delay each have a complexity of $O(\ell (n \ell^{3k+1})^2)=O(\ell^{6k+3}n^2)$. While this is polynomial in $n$ and $\ell$, we can lower the degree by exploiting the structure of $\Aeq$.
	
	Take particular note of the structural similarity between $\Aeq$ and the graphs that are used in the proof of Theorem~\ref{thm:delayAlgorithm} when matching on the string $\strs$: Both have $\ell$ levels, where a level $i$ represents that $i$ letters of $\ell$ have been processed.
	Hence, if we wanted to enumerate $\spnr{\Aeq}(\strs)$ as described in the proof of Theorem~\ref{thm:delayAlgorithm}, the constructed graph would only need to have $O(\ell^{3k+1})$ nodes, with exactly one outgoing edge for all nodes (except the initial and the final node). Consequently, the resulting NFA that is used for $\enumalgo$ would only have $O(\ell^{3k+1})$ nodes and $O(\ell^{3k+1})$ transitions.
	
	We now examine $\Ajoin$. Each of its states $(q,\qeq)$ encodes a state $q$ of $A$ and a state $\qeq$ of $\Aeq$. The state $q$ has $O(n)$ successor states in $A$, and $\qeq$ has at most one successor state in $\Aeq$ (unless it is the initial state; but as this state cannot be reached after it has been left, the number of transitions from it are dominated by the number of the other transitions). Hence, $(q,\qeq)$ has at most $O(n)$ successors in $\Ajoin$. In total, $\Ajoin$ has $O(\ell^{3k+1}n)$ states, and $O(\ell^{3k+1}n^2)$ transitions. By combining these observations with the construction from the proof of Lemma~\ref{lem:joinVset}, we see that $\Ajoin$ can be constructed in time $O(vn^2 + mn + \ell^{3k+1} n^2)$.
	
	Moreover, $\Ajoin$ exhibits the same level structure as $\Aeq$. Hence, in order to enumerate $\spnr{\Ajoin}(\strs)$, we can directly derive $G$ and $A_G$ from $\Ajoin$, without needing the extra factor of $O(\ell)$ states to store how much of $\strs$ has been processed. 
	
	Thus,  $A_G$ has $O(\ell^{3k+1}n)$ states and $O(\ell^{3k+1}n^2)$ transitions. In time $O(\ell^{3k+1} n^2)$, we can construct $A_G$ from $\Ajoin$, precompute its functions $\minlet$ and $\nextlet$, and ensure that the final state is reachable from every state, and that each state is reachable. Finally, note that the alphabet $\Konfset$ of $A_G$ is the set of all variable configurations of $A$, as $\svars(\Aeq)\subseteq \svars(A)$.
	
	All that remains is calling $\enumalgo$ on $A_G$. As in the proof of Theorem~\ref{thm:delayAlgorithm}, its complexity is determined by the complexity of $\minword$: The for-loop is executed $O(\ell)$ times. As each level of $A_G$ contains $O(\ell^{3k}n)$ nodes, the same bound holds for each $S_i$. Furthermore, as each of these nodes has $O(n)$ successors, each $S_{i+1}$ can be computed in time $O(\ell^{3k}n^2)$. Hence, executing $\minword$ requires time $O(\ell^{3k+1}n^2)$. 
	
	Hence, we can enumerate $\spnr{\Ajoin}(\strs)$ with a delay of $O(\ell^{3k+1}n^2)$. Including the construction of $\Ajoin$, the preprocessing takes $O(vn^2 + mn + \ell^{3k+1}n^2)$, which we can simplify to $O(n^3 + \ell^{3k+1}n^2)$. Except for rather pathological cases, we can assume this to be $O(\ell^{3k+1}n^2)$. 	
\end{proof}	
	
\end{document}